\documentclass[10pt]{article}
\usepackage{amsmath,amssymb,amsthm,amsfonts,amsbsy,latexsym,bbm}
\usepackage[letterpaper,margin=1.2in]{geometry}
\usepackage[english]{babel}
\usepackage{hyperref}

\newtheorem{theorem}{Theorem}[section]
\newtheorem{lemma}[theorem]{Lemma}
\newtheorem{assumption}[theorem]{Assumption}
\newtheorem*{notation}{Notation}

\newtheorem{remark}[theorem]{Remark}
\newtheorem{definition}[theorem]{Definition}

\renewcommand{\phi}{\varphi}
\renewcommand{\Re}{\mathrm{Re}}
\renewcommand{\Im}{\mathrm{Im}}

\newcommand{\laa}{\langle}
\newcommand{\raa}{\rangle}
\newcommand{\LZ}{L^2(\mathbb{R}^3) }
\newcommand{\LZN}{L^2(\mathbb{R}^{3N}) }
\newcommand{\mlf}{\widehat{m}^{\phi}}
\newcommand{\RRR}{\mathbb{R}}
\newcommand{\NNN}{\mathbb{N}}
\newcommand{\Hilbert}{\mathcal{H}}
\newcommand{\refe}{\mathrm{ref}}
\newcommand{\Bog}{\mathrm{Bog}}
\newcommand{\sym}{\mathrm{sym}}
\newcommand{\mf}{\mathrm{mf}}
\newcommand{\norm}[2][]{\left|\left| #2 \right|\right|_{#1}}
\newcommand{\scp}[2]{\langle #1 , #2 \rangle}
\newcommand{\bigscp}[2]{\Big\langle #1 , #2 \Big\rangle}
\newcommand{\hc}{\mathrm{h.c.}}

\begin{document}

\title{Derivation of the Bogoliubov Time Evolution for a Large Volume Mean-field Limit}

\date{December 11, 2019}

\author{
S\"oren Petrat\footnote{Jacobs University, Department of Mathematics and Logistics, Campus Ring 1, 28759 Bremen, Germany. \newline E-mail: {\tt s.petrat@jacobs-university.de}},
Peter Pickl\footnote{Mathematisches Institut, Ludwig-Maximilians-Universit\"at, Theresienstr.\ 39, {80333} M\"unchen, Germany; and \newline Duke Kunshan University, No.\ 8 Duke Avenue, Kunshan, Jiangsu Province, 215316, China. E-mail: {\tt pickl@math.lmu.de}},
Avy Soffer\footnote{Department of Mathematics, Rutgers University, 110 Frelinghuysen Road, Piscataway, NJ 08854, USA; and \newline School of Mathematics and Statistics, Central China Normal University, No.\ 152 Luoyu Road, Wuhan Hubei, 430079, China. E-mail: {\tt soffer@math.rutgers.edu}}
}

\maketitle

\begin{abstract}
The derivation of mean-field limits for quantum systems at zero temperature has attracted many researchers in the last decades. Recent developments are the consideration of pair correlations in the effective description, which lead to a much more precise description of both spectral properties and the dynamics of the Bose gas in the weak coupling limit. While mean-field results typically lead to convergence for the reduced density matrix only, one obtains norm convergence when considering the pair correlations proposed by Bogoliubov in his seminal 1947 paper. In this article we consider an interacting Bose gas in the case where both the volume and the density of the gas tend to infinity simultaneously. We assume that the coupling constant is such that the self-interaction of the fluctuations is of leading order, which leads to a finite (non-zero) speed of sound in the gas. In our first main result we show that the difference between the $N$-body and the Bogoliubov description is small in $L^2$ as the density of the gas tends to infinity and the volume does not grow too fast. This describes the dynamics of delocalized excitations of the order of the volume. In our second main result we consider an interacting Bose gas near the ground state with a macroscopic localized excitation of order of the density. We prove that the microscopic dynamics of the excitation coming from the $N$-body Schr\"odinger equation converges to an effective dynamics which is free evolution with the Bogoliubov dispersion relation. The main technical novelty are estimates for all moments of the number of particles outside the condensate for large volume, and in particular control of the tails of their distribution.
\end{abstract}

\noindent
\textbf{MSC class:} 35Q40, 35Q55, 81Q05, 82C10

\section{Introduction}
The effective description of the dynamics of Bose gases has been extensively discussed in the mathematical physics literature, see \cite{schleinbook} for an overview of the topic, and \cite{hepp:1974,GinibreVelo:1979a,GinibreVelo:1979b,spohn,erdos6,FrohlichKnowlesSchwarz:2009,RodnianskiSchlein:2007,ChenLee:2010,pickl,Chen:2011,pickl:2010hartree,Luhrmann:2012,Michelangeli:2012,LewinNamRougerie:2014,anapolitanos:2016} and \cite{erdos1,erdos2,erdos3,erdos4,erdos5,pickl:2010gp_pos,pickl2,benedikter:2012,chen:2014,chenholmer:2016,jeblick:2016,jeblick:2018,rademacher:2019} for a non-exhaustive list of references on the mean-field and NLS limits. Typically the convergence of the many-body system to the effective descriptions is in terms of reduced density matrices. However, considering pair correlations in the gas, it has been shown that one gets $L^2$ convergence in the weak coupling limit, i.e., the limit where the volume $\Lambda$ is kept fixed and the density $\rho$ scales with the particle number $N\sim\rho$ while the coupling constant of the interaction is $N^{-1}$. The idea that pair correlations give a very good description of the ground state and time evolution of the interacting Bose gas goes back to Bogoliubov \cite{Bogoliubov:1947}.

The rigorous analysis of spectral low energy properties in terms of Bogoliubov theory for the weakly interacting Bose gas \cite{lsy,lssy} has been initiated more recently. In \cite{LiebSolovej:2001,LiebSolovej:2004,Solovej:2006,GiulianiSeiringer:2009,
YauYin:2009}, the next-to-leading order contribution $E^{Bog}$ in the ground state energy $E_N^0 = Ne^{(0)} + E^{Bog} + o_N(1)$ has been derived. Then, in \cite{Seiringer:2011}, the complete Bogoliubov theory (of the low energy spectrum and low energy eigenfunctions) was derived for the homogeneous gas on the torus, which was generalized in \cite{Grech:2013,Lewin:2015b}, and further generalized to a mean-field large volume limit in \cite{DerezinskiNapiorkowski:2014}. The first rigorous results on the time-dependent problem were obtained in two papers by Grillakis, Machedon and Margetis \cite{GrillakisMachedonMargetis:2010_1,GrillakisMachedonMargetis:2010_2}. Many papers followed \cite{GrillakisMachedon:2013,Lewin:2015a,Boccato:2016,GrillakisMachedon:2017,nam:2015,nam:2016,namnap_review,namnap_low_dim,mpp,brennecke:2017}, where Bogoliubov theory was derived for the time-dependent problem in the mean-field or NLS limit.

In this paper we extend these results in the direction that we show that the Bogoliubov approximation is also valid for the dynamics of gases of large volumes, which can be regarded as a step in the direction of a thermodynamic limit. To this end we consider a Bose gas with initial conditions supported in a box of volume $\Lambda$ and such that the support essentially remains of order $\Lambda$. In contrast to most works on the mean-field limit we consider a system where both the density and volume are large. At the same time we take the coupling constant for the interaction proportional to $\rho^{-1}$. Such a limit has been introduced by Derezi{\'n}ski and Napi{\'o}rkowski in \cite{DerezinskiNapiorkowski:2014}. Let us assume for the moment that the initial $N$-body wave function $\Psi_N^0 \in L^2(\RRR^{3N})$ is a condensate, i.e., it factorizes as $\Psi_N^0=\phi_0^{\otimes N}$ with some $\varphi_0 \in L^2(\RRR^3)$. We then prove that the wave function $\Psi_N^t$ evolving with the microscopic Schr\"odinger dynamics is still approximately a product state. More precisely, we can show that all but order $\Lambda$ particles are still in the product state, if the same holds for the initial data. Let us call the particles not in the product state the excitations around the condensate.

The first main result of this article is that as long as $\Lambda\ll\rho^{1/3}$, the wave function $\Psi_N^t$ describing the microscopic system converges in $L^2$-norm to a wave function coming from the Bogoliubov approximation on Fock space, i.e., where pair correlations are considered, and we provide explicit bounds. Here, the excitations consist of order $\Lambda$ particles, and are delocalized over the whole system.

Our second main result is an improvement of the main result from \cite{dfpp}. Let us assume, as in \cite{dfpp}, that $\phi_0=\phi^{(\refe)}_0+\epsilon_0$ is the sum of a reference state $\phi^{(\refe)}_0$ and a local excitation $\epsilon_0$. The volume $\Lambda$ enters as the support of $\phi^{(\refe)}_0$, and $\phi^{(\refe)}_0$ is assumed to be a constant except with decay near the boundary. The number of particles building the local excitation is given by the density $\rho$ of the gas, i.e., $N \|\epsilon_0\|^2=\rho$ and $\Lambda=N/\rho$. The local excitation is thus \emph{macroscopic}. It is precisely the introduction of the parameter $\Lambda$ that allows us to consider such a macroscopic excitation. Note that the number of particles in the local excitation compared to the full number of particles goes to zero in our case, but with a rate $\rho N^{-1} = \Lambda^{-1}$ which is much slower than $N^{-1}$. Other than \cite{dfpp}, we are not aware of other similar results where the number of excitations becomes unbounded with $N$ going to infinity. Therefore, we need to improve on existing methods and results where the dependence on $\Lambda$ is not tracked. Using only the mean-field Hartree equation for the evolution of $\phi_0$ and using that the reference state $\phi^{(\refe)}_t$ is almost constant throughout the system, one can track the evolution of the local excitation $\epsilon_t$. In this context one can show that the evolution of $\epsilon_t$ with initial conditions small in $L^2$ is close to a free evolution with dispersion relation $\omega(k)$ given by Bogoliubov theory, i.e.,
\begin{equation}
\omega(k) = \sqrt{k^4 + 2 k^2 \hat{v}(k)},
\end{equation}
where $\hat{v}$ denotes the Fourier transform of the interaction potential $v$, see Section~\ref{sectionsound} and \cite{dfpp} for more details. So in our case the number of excited particles is of order $\rho$, which means---considering the scaling of the coupling constant---that the interaction between the particles forming the excitation is of order one per particle. This guarantees that there is a finite speed of sound when considering the motion of the excitation. Our main result is that the effective dynamics of $\epsilon_t$ is indeed close to the microscopic dynamics coming from $\Psi_N^t$ as long as $\Lambda \ll \rho$. This means that as long as the $\epsilon_0$ excitation is made up of many more particles than there are particles outside the condensate, we can track its dynamics effectively. This improves the result from \cite{dfpp}, where $\Lambda^3 \ll \rho$ was assumed.

The article is structured in the following way. We present and discuss our main results in Section~\ref{sec_main_results}. In Section~\ref{sec_main_results_counting} we explain the key technical result Lemma~\ref{cor}, that makes the proofs of the validity of Bogoliubov's approximation for large volume possible. Then Section~\ref{section:Bog} is about the norm convergence to the Bogoliubov dynamics, while Section~\ref{sectionsound} is about the result for the local density excitation. The proofs are given in Sections~\ref{sec_proof_Hartree_stuff} and \ref{sec_proof_the_beyond}. In Sections~\ref{sec_proof_counting} and \ref{section_alpha_der_gen} we derive a general estimate for the change in the number of particles outside the condensate. The main result here is Lemma~\ref{alphaabl}, which contains a very general bound with the freedom of choosing a certain weight function. In Section~\ref{sec_proof_cor} we choose a particular weight function that enables us to prove Lemma~\ref{cor}. This is the main technical novelty here; it allows us to control the large tails in the distribution of the number of particles outside the condensate. Finally, the main results are proven in Section~\ref{proofs_thms}.

\section{Main Results}\label{sec_main_results}

\subsection{The Number of Particles Outside the Condensate}\label{sec_main_results_counting}

We are interested in solutions of the $N$-particle Schr\"odinger
equation
\begin{equation}\label{schroe}
i\partial_t \Psi_N^t = H_N\Psi_N^t
\end{equation}
for symmetric initial conditions $\Psi_N^0$ and with Hamiltonian
\begin{equation}\label{hamiltonian}
H_N=-\sum_{j=1}^N \Delta_j+\sum_{1\leq j< k\leq N}\rho^{-1}  v(x_j-x_k)
\end{equation}
acting on the Hilbert space $\LZN$, with $v:\RRR^3\to\mathbb R^+_0$, parameter $\rho > 0$, and where we have set $\hbar= 1 = 2m$. In our main theorems, we will introduce another parameter $\Lambda > 0$ through the initial conditions in such a way that one can identify $\Lambda$ with the volume of the system. One should think of $\Lambda$ as the support of the initial conditions, which remains essentially unchanged in finite time. Since we choose $v \in L^1(\RRR^3) \cap L^{\infty}(\RRR^3)$ and $v>0$, $v$ has finite scattering length \cite{LiebYngvason}, and we express all lengths with reference to this scattering length. This allows us to speak of ``large volume''. From now on we will call $\Lambda$ \emph{the volume} and $\rho := N/\Lambda$ \emph{the density}. In the context of mean-field Bose gases, the scaling limit of \eqref{hamiltonian} has to our knowledge first been introduced in \cite{DerezinskiNapiorkowski:2014} to derive the Bogoliubov excitation spectrum. If one does not keep track of the volume, the mean-field limit corresponds to replacing $\rho^{-1}$ by $N^{-1}$ in \eqref{hamiltonian}.

For the kind of interaction potentials $v$ we consider in our main results ($v \in L^{\infty}(\RRR^3)$ in particular), the Hamiltonian $H_N$ is self-adjoint on $H^2(\RRR^{3N})$, the second Sobolev space. It is thus the generator of the unitary group $U(t) = e^{-iH_Nt}$. $H_N$ conserves symmetry under exchange of particle labels, i.e., any symmetric function $\Psi_N^0$ evolves into a symmetric function $\Psi_N^t$. Throughout this article we assume $N\geq 2$.

Let us also introduce the Hartree equation
\begin{equation}\label{meanfield}
i\partial_t
\phi_t=\left(-\Delta+\frac{N-1}{\rho}\big[(v*|\phi_t|^2)-\mu^{\varphi_t}\big] \right)\phi_t:=h^{\varphi_t}\phi_t
\end{equation}
with initial condition $\phi_0 \in L^2(\RRR^3)$, where $\mu^{\varphi_t} = \frac{1}{2} \int (v*|\phi_t|^2) |\phi_t|^2$ and $*$ means convolution. As we will show, $\phi_t$ effectively describes the time evolution of the Bose gas. To get some intuition about the size of the kinetic and potential terms in the Hartree equation, think of a plane wave with momentum of order $\Lambda^{-1/3}$. Then the kinetic energy per particle is of order $\Lambda^{-2/3}$. On the other hand, since we consider short range $v\in L^1(\RRR^3)$, the potential energy per particle is of order one. While this disparity means that the solution of the Hartree equation is not so interesting, the dynamics of the fluctuations around the mean-field is, which is what we study in this article.

Next, we define the projectors on one-body states $\phi$ and their complement.

\begin{definition}\label{defproa}
Let $\phi\in\LZ$ with $\norm{\varphi} = 1$. For any $1\leq j\leq N$ the
projectors $p_j^\phi:\LZN\to\LZN$ and $q_j^\phi:\LZN\to\LZN$ are given by
\begin{equation}
(p_j^\phi\Psi_N)(x_1,\ldots,x_N):=\phi(x_j)\int\overline{\phi(x_j)}\Psi_N(x_1,\ldots,x_N)dx_j\;\;\;\forall\;\Psi_N\in\LZN,
\end{equation} 
where $\overline{\varphi}$ denotes complex conjugation, and $q_j^\phi:=1-p_j^\phi$. We shall also use the bra-ket notation
$p_j^\phi=|\phi(x_j)\rangle\langle\phi(x_j)|$.
\end{definition}

In order to make the notion of ``particles outside the condensate'' more precise, note that any wave function $\Psi_N \in L^2(\RRR^{3N})$ can be decomposed as
\begin{equation}\label{decomp}
\Psi_N = \sum_{k=0}^N \varphi^{\otimes (N-k)} \otimes_s \chi^{(k)},
\end{equation}
where $\otimes_s$ denotes the symmetric tensor product, and the $k$-particle wave functions $\chi^{(k)} \in L^2(\RRR^{3k})$ are chosen orthogonal to $\varphi$ in each tensor component. To our knowledge, this decomposition appeared first in the context of deriving the Bogoliubov excitation spectrum in \cite{Lewin:2015b} and the dynamics in \cite{Lewin:2015a}. We have thus decomposed $\Psi_N$ into a sum of wave functions with exactly $k$ particles outside the condensate. We will often call $\chi^{(k)}$ the $k$-particle excitations. The relation \eqref{decomp} can be inverted, yielding
\begin{equation}
\chi^{(k)}(x_1,\ldots,x_k) = \sqrt{\binom{N}{k}} \prod_{i=1}^k q^{\varphi}_i \int \prod_{i=k+1}^N\overline{\varphi(x_i)} \Psi_N(x_1,\ldots,x_N) dx_{k+1}\ldots dx_N.
\end{equation}
Note in particular that
\begin{equation}
\scp{\Psi_N}{q^{\varphi}_1\Psi_N} = \sum_{k=0}^N \frac{k}{N} \big\|\chi^{(k)}\big\|^2
\end{equation}
by using the symmetry of $\Psi_N$, i.e., $\scp{\Psi_N}{q^{\varphi}_1\Psi_N}$ is the relative number of particles not in the condensate state $\varphi$.

Before one can prove the validity of the Bogoliubov equations, one first has to prove bounds on the number of particles outside the condensate. The key step is Lemma~\ref{cor}, which in particular implies
\begin{equation}\label{q1estimate}
\sum_{k=0}^N \frac{k}{N} \big\|\chi^{(k)}_t\big\|^2 \leq C(t) \rho^{-1}
\end{equation}
for some positive $C(t)$ (that can be inferred from Lemma~\ref{cor}), and which holds for $1 \leq \Lambda \leq \rho^{1-\varepsilon}$, for any $0<\varepsilon<1$. Note that the $\rho^{-1}$ on the right-hand side is expected to be optimal. For the Bogoliubov time evolution or a related auxiliary time evolution that we introduce in Section~\ref{sec_proof_the_beyond}, such a bound can be proven with small modifications of existing techniques. However, for the full $N$-body evolution, the best previously known bound similar to \eqref{q1estimate} was, as far as we know, given in \cite{dfpp} and reads
\begin{equation}
\sum_{k=0}^{\rho} \frac{k}{N} \big\|\chi^{(k)}_t\big\|^2 \leq C(t) \rho^{-1}.
\end{equation}
The difference to our bound \eqref{q1estimate} is that the ``large tail'' of the distribution $\big\|\chi^{(k)}_t\big\|^2$ could not be suitably controlled. This is the main technical innovation of our article. The bound \eqref{q1estimate} will allow us to prove optimal rates for the convergence to the Bogoliubov time evolution (optimal in the current setting, cf. Remark~\ref{new_label_remark}). Note that the estimate \eqref{q1estimate} does not quite allow us to take the limit $\Lambda \to \infty$ for fixed $\rho$ due to the constraint $\Lambda \leq \rho^{1-\varepsilon}$. 

Bounds of the type \eqref{q1estimate} are standard in the derivation of mean-field dynamics when the dependence on $\Lambda$ is not tracked, see, e.g., \cite{RodnianskiSchlein:2007,pickl,Lewin:2015a}. However, following the proofs in \cite{RodnianskiSchlein:2007,pickl,Lewin:2015a} and just keeping track of the parameter $\Lambda$ would lead to a bound with $C(t)$ of the form $e^{C\sqrt{\Lambda}t}$, which is far from optimal.

We prove an estimate of the type \eqref{q1estimate} also for higher moments. By Lemma~\ref{cor}, we have for $\ell \in \NNN$,
\begin{equation}\label{estimate_from_corollary}
\sum_{k=0}^N \left(\frac{k}{N}\right)^{\ell} \norm{\chi^{(k)}_t}^2 \leq C_{\ell}(t) \rho^{-\ell}
\end{equation}
for some positive $C_{\ell}(t)$, which again holds for $1 \leq \Lambda \leq \rho^{1-\varepsilon}$, for any $0<\varepsilon<1$. Lemma~\ref{cor} is the key step to prove the validity of the Bogoliubov approximation in two different settings: for the $O(\Lambda)$ number of particles outside the condensate, see Section~\ref{section:Bog}, and for the dynamics of a density excitation made up of $O(\rho)$ particles, see Section~\ref{sectionsound}. Note also that the estimate \eqref{q1estimate} implies convergence of the one-body reduced density matrix of $\Psi_N^t$ to $p^{\varphi_t}$.

\subsection{Norm Convergence and the Bogoliubov Approximation}\label{section:Bog}

In this article we go beyond the mean-field description and compare the time evolution of $\Psi_N$ with the Bogoliubov time evolution. Let us use the notation $h^{\varphi_t}_i$ for the Hartree operator from \eqref{meanfield} acting on coordinate $i$ and define the Fock space vector 
\begin{equation}
\chi^{\Bog}_t = \big( \chi^{\Bog,(0)}_t,\chi^{\Bog,(1)}_t,\chi^{\Bog,(2)}_t,\ldots \big) ~~\text{with}~~ \big\| \chi_0^{\Bog} \big\|^2 := \sum_{k=0}^{\infty} \big\| \chi_0^{\Bog,(k)} \big\|^2.
\end{equation}
 We assume that its components satisfy the coupled equations
\begin{align}\label{chi_Bog_k}
i\partial_t \chi^{\Bog,(k)}_t &= \sum_{i=1}^k \bigg(h^{\varphi_t}_i + K^{(1)}_t(x_i)\bigg) \chi^{\Bog,(k)}_t \nonumber\\
&\quad + \frac{1}{2} \frac{1}{\sqrt{k(k-1)}} \sum_{1 \leq i < j \leq k} K^{(2)}_t(x_i,x_j) \chi^{\Bog,(k-2)}_t(x_1,\ldots,x_k \setminus x_i \setminus x_j) \nonumber\\
&\quad + \frac{1}{2} \sqrt{(k+1)(k+2)} \int dx \,dy\, \overline{K^{(2)}_t(x,y)} \chi^{\Bog,(k+2)}_t(x_1,\ldots,x_k,x,y),
\end{align}
where $(x_1,\ldots,x_k \setminus x_i \setminus x_j)$ means the configuration $(x_1,\ldots,x_k)$ with $x_i$ and $x_j$ removed, and
\begin{equation}\label{sys1}
K^{(1)}_t: \Hilbert \to \Hilbert,~ K^{(1)}_t = q^{\varphi_t}\widetilde{K}^{(1)}_tq^{\varphi_t} ~\text{with}~ \widetilde{K}^{(1)}_t(x,y) = \Lambda\varphi_t(x)v(x-y)\overline{\varphi_t(y)},
\end{equation}
\begin{equation}\label{sys2}
K^{(2)}_t \in \Hilbert^{\otimes 2},~ K^{(2)}_t = q^{\varphi_t}\otimes q^{\varphi_t}\widetilde{K}^{(2)}_t ~\text{with}~ \widetilde{K}^{(2)}_t(x,y) = \Lambda v(x-y)\varphi_t(x)\varphi_t(y),
\end{equation}
where $\Hilbert = L^2(\RRR^3)$. Written in terms of second quantization, $\chi^{\Bog}_t$ satisfies
\begin{equation}\label{Bog_eq}
i\partial_t \chi^{\Bog}_t = H^{\Bog}_t \chi^{\Bog}_t
\end{equation}
with
\begin{align}\label{Bogoliubov_a}
H^{\Bog}_t &= \int a^{\dagger}_x \Big(h^{\varphi_t}_x + K^{(1)}_t(x)\Big) a_x dx + \frac{1}{2} \int\int \Big( K^{(2)}_t(x,y) a^\dagger_x a^\dagger_y + \overline{K^{(2)}_t(x,y)} a_x a_y \Big)dx\,dy,
\end{align}
where $a^\dagger,a$ are bosonic creation and annihilation operators with the canonical commutation relations (CCR)
\begin{equation}
[a_x,a^\dagger_y]=\delta_{xy}, ~~~~ [a^\dagger_x,a^\dagger_y] = 0 = [a_x,a_y].
\end{equation}
The Hamiltonian \eqref{Bogoliubov_a} is a quadratic operator on Fock space called the Bogoliubov Hamiltonian. Equation \eqref{Bog_eq} is usually called Bogoliubov-de Gennes equation. Note that the Bogoliubov Hamiltonian creates and annihilates only pairs of particles. We have defined the Bogoliubov Hamiltonian in such a way that it cancels the leading order in the time evolution of the excitations around the mean-field. This is the main result of this section.
\begin{theorem}\label{main_Bog_thm}
Let $\Lambda\geq 1$, $T>0$, $v \in L^1(\RRR^3) \cap L^{\infty}(\RRR^3)$ non-negative, and $1 = \big\| \chi_0^{\Bog} \big\| = \|\Psi_N^0\| = \|\varphi_0\|$. Assume
\begin{equation}\label{cond_on_ini_cond_main_in_thm}
\sum_{k=0}^N \left(\frac{k}{N}\right)^{\ell} \norm{\chi_0^{\Bog,(k)}}^2 \leq c_{\ell} \rho^{-\ell}
\end{equation}
for all $1 \leq \ell \leq 4$ and for some $c_{\ell} \in \RRR$. Also assume that there is a constant $M<\infty$, uniform in $\Lambda$ and $\rho$, such that the solution $\phi_t$ of the Hartree equation \eqref{meanfield} satisfies $\sup_{0\leq t \leq T}\|\phi_t\|_\infty<M\Lambda^{-1/2}$. Let $\Psi_N^t$ be the solution to \eqref{schroe} and $\chi^{\Bog}_t$ to \eqref{Bog_eq}. Then for all $t\leq T$ there is a positive $C(t)$ (uniform in $\Lambda$ and $\rho$) such that
\begin{equation}\label{Bog_bound_psi_thm}
\bigg\|\Psi_N^t - \sum_{k=0}^N \varphi_t^{\otimes (N-k)} \otimes_s \chi^{\Bog,(k)}_t\bigg\| \leq 6 \sqrt{\bigg\|\Psi_N^0 - \sum_{k=0}^N \varphi_0^{\otimes (N-k)} \otimes_s \chi^{\Bog,(k)}_0\bigg\|} + C(t) \sqrt{\frac{\Lambda^3}{\rho}}.
\end{equation}
\end{theorem}

The proof is given in Section~\ref{proofs_thms}.

\begin{remark}
~
\begin{itemize}
\item[(a)] The function $C(t)$ can be chosen as 
\begin{equation}
C(t) = 4C_4^{1/2}\big(e^{D_4t}-1\big)
\end{equation}
with
\begin{equation}\label{constants_from_thm}
D_4 = 3^4 \,36 (1+M)^2 (\|v\|_1+\|v\|_2+\|v\|_\infty), \quad C_4 = 1 + \sum_{m=1}^4 (2^m c_m + 1).
\end{equation}
\item[(b)] In \cite{dfpp} it has been shown that for positive (i.e., repulsive) smooth $v$ with compact support, and for initial data as in Assumption~\ref{main_assumption} below, there is an $M>0$ (uniform in $\Lambda$ and $\rho$) and a solution $\varphi_t$ of the Hartree equation \eqref{meanfield} that satisfies $\sup_{0\leq t \leq T}\norm[\infty]{\varphi_t} \leq M \Lambda^{-1/2}$ for any $T>0$. In other words, if we assume that Assumption~\ref{main_assumption} holds, then Theorem~\ref{main_Bog_thm} is true only under assumptions on the initial data. Note that the technical assumptions of smoothness and compact support can possibly be relaxed, but that for attractive interactions the solution might actually blow up in finite time, so the estimate for the infinity norm and thus our result might only hold for finite time. Note also that $M$ is not necessarily uniform in $T$, so $C(t)$ might grow faster than exponential in time.
\item[(c)] The assumptions on the initial conditions in \eqref{cond_on_ini_cond_main_in_thm} say that initially the forth and lower moments of the number of particles outside the condensate have to be of the same order in $\rho$ as the growth in time, see \eqref{estimate_from_corollary} or Lemma~\ref{cor}. In particular, the assumptions hold for product states.
\item[(d)] Note that the state $\sum_{k=0}^N \varphi_t^{\otimes (N-k)} \otimes_s \chi^{\Bog,(k)}_t$ is typically not normalized, but $\sum_{k=0}^{\infty} \big\| \chi_t^{\Bog,(k)} \big\|^2 = 1$ since the Bogoliubov time evolution is unitary.
\item[(e)] For the kind of $\varphi_t$ and $v$ we consider in our result, well-posedness of the Bogoliubov equations can be established similarly to \cite{Lewin:2015a,bbcfs}.
\end{itemize}
\end{remark}

Arguably the most important property which makes \eqref{Bog_eq} such a useful approximation is that it preserves quasi-free states. This means, that all correlation functions can be expressed in terms of two-point functions only. Here, these are the reduced one-body density matrix $\gamma_t: \Hilbert \to \Hilbert$ and the pairing density $\alpha_t: \Hilbert^* \to \Hilbert$, defined by  
\begin{equation}
\scp{f}{\gamma_t g} = \scp{\chi_t^{\Bog}}{a^\dagger(g) a(f) \chi_t^{\Bog}}, ~~~ \scp{f}{\alpha_t \overline{g}} = \scp{\chi_t^{\Bog}}{a(g) a(f) \chi_t^{\Bog}}
\end{equation}
for all $f,g \in \Hilbert$, where 
\begin{equation}
a^\dagger(f) = \int f(x) a_x^\dagger dx, ~~~ a(f) = \int \overline{f(x)} a_x dx.
\end{equation}
Then a direct calculation shows that if $\chi^{\Bog}_t$ solves the Bogoliubov equation \eqref{Bog_eq} with Hamiltonian \eqref{Bogoliubov_a}, then $\gamma_t$ and $\alpha_t$ solve the closed system of equations
\begin{align}\label{gamma_alpha_Bog1}
i \partial_t \gamma_t & =  \Big[ h^{\varphi_t} + K^{(1)}_t, \gamma_t \Big] + {K^{(2)}_t} \alpha_t^{\dagger} - \alpha_t \left(K^{(2)}_t\right)^{\dagger}, \\
\label{gamma_alpha_Bog2} i \partial_t \alpha_t & =  \big( h^{\varphi_t} + K^{(1)}_t \big) \alpha_t + \alpha_t \big( h^{\varphi_t} + K^{(1)}_t \big)^{\text T} + K^{(2)}_t +  K^{(2)}_t \gamma_t^{\text{T}} + \gamma_t K^{(2)}_t,
\end{align}
where $\gamma^T:\Hilbert^* \to \Hilbert^*$ is the operator with kernel $\gamma_t^T(x,y) = \gamma_t(y,x)$. Therefore, if we choose an initial state that is close to a quasi-free state, we can solve the Bogoliubov equation \eqref{Bog_eq} simply by solving the two equations \eqref{gamma_alpha_Bog1} and \eqref{gamma_alpha_Bog2} on $\RRR^6$ (given the solution of the Hartree equation \eqref{meanfield}). Let us refer to \cite{nam:2015,nam:2016} for a more thorough discussion. Let us also note that by projecting the solutions of the $N$-body Schr\"odinger equation \eqref{schroe} to quasi-free states, i.e., by \emph{assuming} the quasi-free property, one arrives at the Bogoliubov time evolution \eqref{Bog_eq}, see, e.g., \cite{bbcfs}. 

The Bogoliubov Hamiltonian \eqref{Bogoliubov_a} can be diagonalized under quite general conditions by using the correspondence between evolution operators generated by quadratic Hamiltonians and Bogoliubov transformations, see \cite{namnapjps,napior:2018} for recent results. For translation invariant initial data, the result is
\begin{equation}
H^{\Bog} = E + \sum_k \omega(k) b^{\dagger}_k b_k
\end{equation}
for some constant $E\in \RRR$ and new creation and annihilation operators $b^{\dagger}_k$, $b_k$ coming from a linear transformation of $a^{\dagger}_k$, $a_k$. The dispersion relation is $\omega(k) = \sqrt{k^4 + 2 k^2 \hat{v}(k)}$. Defining the speed of sound as the derivative of the dispersion relation at zero, we find $\partial_k \omega(k)|_{0} = \sqrt{2\hat{v}(0)}$ for the speed of sound in a translation invariant Bose gas.

\subsection{The Speed of Sound in the Bose Gas}\label{sectionsound}

In \cite{dfpp} the mean-field limit for a gas of large volume is considered. The estimates therein are sufficiently strong to be able to track the dynamics of an initially localized excitation in the Bose gas of large volume. For this section only let us make two changes in our definitions in order to adopt the notation from \cite{dfpp} and thus make our results easier to compare, and such that Theorem~\ref{side} can be stated in a nicer way. First, we choose $\varphi_0$ such that $\|\varphi_0\|=\Lambda^{1/2}$ instead of $1$, see Assumption~\ref{main_assumption} below. Second, we define $\mu^{\varphi_t}=0$ in the Hartree equation \eqref{meanfield}. The conditions for the initial wave function are summarized in the following assumption. This is Condition 1.2 in \cite{dfpp}, except that we changed the wording a bit and removed the assumption of exact factorization of the initial state $\Psi_N^0$.
\begin{assumption}[Slightly modified Condition 1.2 in \cite{dfpp}]\label{main_assumption}
The initial conditions are
\begin{equation}
\Psi_N^0 \in L^2(\RRR^{3N}) ~~\text{and}~~ L^2(\RRR^3) \ni \varphi_0 = \varphi_0^{(\refe)} + \epsilon_0
\end{equation}
with $\norm{\Psi_N^0} = 1$ and $\norm{\varphi_0} = \Lambda^{1/2}$. For some $C>0$, the functions $\phi_{0}^{(\mathrm{ref})},\epsilon_0 \in{\cal C}_{c}^{\infty}(\RRR^3)$ (infinitely often differentiable with compact support) have the following properties:
\begin{align}\label{eq:conditions-phi}
\operatorname{supp}\phi_{0}^{(\mathrm{ref})}\subseteq\Lambda, \qquad
\Vert\phi_{0}^{(\mathrm{ref})}\Vert_\infty \leq \left\Vert \,\widehat{\,|\phi_{0}^{(\mathrm{ref})}|} \, \right\Vert_{1} \leq C,
\end{align}
\begin{align}\label{eq:conditions-epsilon}
\operatorname{supp}\epsilon_{0}\subset{\cal B}_{\frac{1}{4} \Lambda^{1/3}}, 
\qquad \Vert \epsilon_{0}\Vert _{\infty} \leq \left\Vert \, \widehat{|\epsilon_0|} \, \right\Vert_1 \leq C,
\qquad \Vert \epsilon_0 \Vert_2 \leq C,
\end{align}
where in the first inclusion, $\Lambda$ refers to a cube in $\RRR^3$ with volume that we otherwise also denote by $\Lambda$. Here ${\cal B}_{\frac{1}{4} \Lambda^{1/3}}$ stands for the ball centered at the origin with radius $\frac{1}{4} \Lambda^{1/3}$. Furthermore, we assume that the density of the gas  is essentially constant in some large region inside the support of $\varphi_0$. Therefore, with the help of a family of cut-off functions $\chi_r \in {\cal C}^2(\mathbb R^3)$, $0< r< 1$,
\begin{align}\label{eq:cutoff}
\chi_r(x)= \begin{cases} 0 & \text{ for }x \in \mathcal B_{r \Lambda^{1/3}} \\ 1 & \text{ for }x \notin \mathcal B_{\Lambda^{1/3}} \end{cases}
\qquad \text{and} \qquad
\Vert\nabla \chi_r\Vert_\infty\leq C \Lambda^{-1/3},
\end{align}
we require
\begin{equation}\label{eq:phi-cutoff-constraint-epsilon-decay} 
\left| \, \phi_0^{(\mathrm{ref})}(x)-1 \, \right|\leq \chi_{1/2}(x).
\end{equation}
This will allow us to track the dynamics of the excitation with the properties \eqref{eq:conditions-epsilon} in that region. Finally, we require some control of the kinetic energy of the initial reference wave function: 
\begin{equation}\label{eq:nabla-phi-2}
\Vert\nabla\phi^{(\mathrm{ref})}_0\Vert_\infty\leq C\Lambda^{-\frac{1}{3}}\;,\hspace{1cm}
\Vert \nabla\phi^{(\mathrm{ref})}_0 \Vert_2\leq C\Lambda^{\frac{1}{6}}\;,\hspace{1cm}
\Vert \Delta\phi^{(\mathrm{ref})}_0 \Vert_2 \leq C\Lambda^{-\frac{1}{6}}.
\end{equation}
\end{assumption}
One should have an initial state $\Psi_N^0\approx\left(\Lambda^{-1/2}\varphi_0\right)^{\otimes N}$ in mind. In particular, in case of equality, $\Psi_N^0$ is normalized. The precise meaning of $\approx$ is given in \eqref{cond_on_ini_cond_exc} below. In the following, $\Psi_N^t$ denotes a solution to the Schr\"odinger equation \eqref{schroe}, and $\varphi_t$ a solution to the Hartree equation \eqref{meanfield} with the definition $\mu^{\varphi_t}=0$, and where we replace $v*|\varphi_t|^2$ by $\Lambda^{-1} v*|\varphi_t|^2$ due to our assumption $\|\varphi_0\|=\Lambda^{1/2}$. As in \cite{dfpp}, we write $\varphi_t = \varphi_t^{(\refe)} + \epsilon_t$, and assume that $\varphi_t^{(\refe)}$ solves its own Hartree equation
\begin{equation}
i \partial_t \varphi_t^{(\refe)} = \left(-\Delta + (v*|\varphi_t^{(\refe)}|^2) - \|v\|_1 \right)\varphi_t^{(\refe)}.
\end{equation}
The excitation $\epsilon_t$ is then defined by subtracting the reference state from $\varphi_t$ with the correct phase taken into account, viz.
\begin{equation}
\epsilon_t := \varphi_t e^{i\|v\|_1 t} - \varphi_t^{(\refe)}.
\end{equation}
One then finds by direct calculation that $\epsilon_t$ is a solution to
\begin{align}\label{epsilon_eq}
i \partial_t \epsilon_t &= \left(-\Delta + (v*|\varphi_t^{(\refe)}|^2) - \|v\|_1 + (v*|\epsilon_t|^2) + \Re (v*(\overline{\epsilon_t}\varphi_t^{(\refe)})) \right) \epsilon_t \nonumber\\
&\quad + \left((v*|\epsilon_t|^2) + 2\Re (v*(\overline{\epsilon_t}\varphi_t^{(\refe)})) \right) \varphi_t^{(\refe)}.
\end{align}
We can now define
\begin{equation}\label{rho_micro}
\rho_{t}^{(\mathrm{micro})}:=q_{t}^{(\mathrm{ref})} \operatorname{Tr}_{x_{2},\ldots,x_{N}} \left|\Lambda^{1/2}\Psi_N^t\right\rangle \left\langle \Lambda^{1/2}\Psi_N^t\right|q_{t}^{(\mathrm{ref})},
\end{equation}
which is the reduced density of $\Lambda^{1/2}\Psi_N^t$ projected onto the orthogonal complement of the reference state $\phi^{(\refe)}_t$, i.e., $q_{t}^{(\mathrm{ref})}=1-\big\|\phi^{(\refe)}_t\big\|^{-2} \,|\phi^{(\refe)}_t\rangle\langle\phi^{(\refe)}_t|$, and the macroscopic reduced density
\begin{equation}\label{rho_macro}
\rho_{t}^{(\mathrm{macro})}:=\left|\epsilon_{t}\right\rangle \left\langle \epsilon_{t}\right|.
\end{equation}
Note that $\rho_{t}^{(\mathrm{macro})}$ has trace $\norm[2]{\epsilon_t}^2$, which can be shown to be bounded by some $C(t)$ independent of $N$ and $\Lambda$ under the conditions of Theorem~\ref{side}. However, the $L^2$-norm of $\epsilon_t$ is not conserved.

For initial data as in Assumption~\ref{main_assumption} and for positive smooth $v$ with compact support it has been shown in Theorem~1.4 in \cite{dfpp} that there exists a continuous, non-decreasing, non-negative $C(t)$ such that
\begin{equation}\label{compare}
\left\Vert \rho_t^{(\mathrm{micro})} - \rho_t^{(\mathrm{macro})} \right\Vert \leq C(t)\frac{\Lambda^{3/2}}{\rho^{1/2}},
\end{equation}
for all times $t\geq 0$ and sufficiently large $\Lambda$. Our estimate \eqref{estimate_from_corollary} (or Lemma~\ref{cor}) applies to this situation and we are able to improve the error estimate compared to \cite{dfpp}. In particular we are able to prove convergence of the reduced densities in situations where $\rho\gg\Lambda\gg1$ compared to $\rho\gg\Lambda^3\gg 1$ as in \eqref{compare}.

\begin{theorem}\label{side}
Let $v\in{\cal C}^\infty_c(\mathbb R^3,\mathbb R^+_0)$, let Assumption~\ref{main_assumption} for the initial conditions hold, and let $1 \leq \Lambda \leq C \rho^{1-\varepsilon}$ for some $C>0$ and some $0 < \varepsilon < 1$. Assume
\begin{equation}\label{cond_on_ini_cond_exc}
\sum_{k=0}^N \left(\frac{k}{N}\right)^{\ell} \norm{\chi^{(k)}_0}^2 \leq c_{\ell} \rho^{-\ell}
\end{equation}
for all $1 \leq \ell \leq \frac{1}{\varepsilon}+1$ and for some $c_{\ell} \in \RRR$, where here the decomposition \eqref{decomp} is defined with respect to $\Lambda^{-1/2} \varphi_0$ to account for the different normalization. Then for all times $t \geq 0$ there exists a continuous, non-decreasing, non-negative $C_{\varepsilon}(t)$ such that 
\begin{equation}\label{micro_macro_main_result}
\norm{\rho_t^{(\mathrm{micro})} - \rho_t^{(\mathrm{macro})}} \leq C_{\varepsilon}(t) \left( \Lambda^{-1} + \Lambda^{-1/2} + \sqrt{\frac{\Lambda}{\rho}} \, \right),
\end{equation}
where $\rho_t^{(\mathrm{micro})}$ and $\rho_t^{(\mathrm{macro})}$ are defined in \eqref{rho_micro} and \eqref{rho_macro}.
\end{theorem}
We give the proof in Section~\ref{proofs_thms}. From a technical point of view, it is the control of the tails of the distribution of the number of particles outside the condensate which allows us to prove our result for $\Lambda \ll \rho$ instead of $\Lambda^3 \ll \rho$. The key step is the estimate \eqref{corzwei} from Lemma~\ref{cor}. It reads
\begin{equation}
\sum_{k=0}^N \left(\frac{k}{N}\right)^{\ell} \norm{\chi^{(k)}_t}^2 \leq C_{\ell} e^{D_{\ell}t} \rho^{-\ell} + C_j e^{D_jt} \left(\frac{\Lambda}{\rho}\right)^j
\end{equation}
for some constants $C_{\ell}$ and $D_{\ell}$, and for any $j\in \NNN$. The $(\Lambda/\rho)^j$ term comes from the tails where $k\geq \rho$, and we can make it small by choosing $j$ large enough.
\begin{remark}\label{new_label_remark}
~
\begin{itemize}
\item[(a)] As stated in the introduction, this result means that as long as the $\epsilon_t$ excitation is made up of many more particles (order $\rho$) than there are particles outside the condensate (order $\Lambda$), we can track its dynamics effectively. This heuristics makes the condition $\rho \gg \Lambda$ plausible. In order to be able to take the limit $\Lambda\to\infty$ while keeping $\rho$ fixed, one would need better control of the spatial distribution of particles outside the condensate, which is very hard to achieve.
\item[(b)] Let us emphasize again that the number of particles outside the condensate defined with respect to $\phi^{(\refe)}_t$ is of order $\rho \| \epsilon_t\|^2$ here, i.e., it is a macroscopic fraction of particles. Therefore, while the norm convergence result Theorem~\ref{main_Bog_thm} still holds when we consider excitations with respect to $\varphi_t$, it does not hold when we define them with respect to $\phi^{(\refe)}_t$ since the assumption \eqref{cond_on_ini_cond_main_in_thm} does not hold in this case.
\end{itemize}
\end{remark}

Equation~\eqref{epsilon_eq} can be simplified further by linearization, i.e., when we consider initial data $\epsilon_0$ that are small in $L^2$. By using that $\varphi_t^{(\refe)}$ is almost constant, neglecting subleading terms, and then linearizing \eqref{epsilon_eq}, one finds that the dynamics of the excitation $\epsilon_t$ is close to a free dynamics with Bogoliubov dispersion law and describes the Goldstone modes in the system. Let us define $\eta_t$ as a solution to the linearized equation
\begin{equation}\label{eta_eq}
i\partial_t \eta_t = -\Delta \eta_t + \Re (v*\eta_t).
\end{equation}
Then for smooth $v$ with compact support it was shown in \cite[Theorem~1.8]{dfpp} that for $\epsilon_0=\eta_0$, $0 \leq t < T$ and large enough $\Lambda$,
\begin{equation}
\|\eta_t-\epsilon_t\|_2 \leq C(T) \left( \Lambda^{-1/6} + \|\epsilon_0\|_2^2 + \|\epsilon_0\|_2^3 \right)
\end{equation}
for some continuous, non-decreasing, non-negative $C(t)$. This result means that $\epsilon_t$ and $\eta_t$ are close in $L^2$ norm for small $\|\epsilon_0\|_2$ (if also $\Lambda$ is large enough). In Fourier space, \eqref{eta_eq} can be written as
\begin{equation}\label{eta_eq_Fourier}
i \partial_t \left(\begin{array}{c} \widehat{\eta}_t(k) \\ \overline{\widehat{\eta}_t(-k)} \end{array}\right) = H(k) \left(\begin{array}{c} \widehat{\eta}_t(k) \\ \overline{\widehat{\eta}_t(-k)} \end{array}\right) \quad\text{with}\quad H(k) = \left(\begin{array}{cc} k^2 + \widehat{v}(k) & \widehat{v}(k) \\ -\widehat{v}(k) & -k^2 - \widehat{v}(k) \end{array}\right).
\end{equation}
Indeed, one finds that $H(k)$ can be diagonalized (although it is not self-adjoint, and in particular the $L^2$ norm of $\eta_t$ is not conserved). Its eigenvalues are $\pm\omega(k)$ with
\begin{equation}
\omega(k) = \sqrt{k^4 + 2 k^2 \widehat{v}(k)},
\end{equation}
which is the familiar Bogoliubov dispersion relation. Since $H(k)^2 = \omega(k)^2$, we can write \eqref{eta_eq_Fourier} in the second order form
\begin{equation}
-\partial_t^2 \widehat{\eta}_t(k) = \Big( k^4 + 2 k^2 \widehat{v}(k) \Big) \widehat{\eta}_t(k).
\end{equation}
For small $k$, this is a wave equation with propagation speed $\sqrt{2\hat{v}(0)}$. In this way, we have derived the speed of sound in the Bose gas for a macroscopic density perturbation. Since $H(k)^2 = \omega(k)^2$, we can also explicitly solve Equation~\eqref{eta_eq_Fourier}. We find
\begin{equation}\label{explicit_eta_solution}
\widehat{\eta}_t(k) = \left( \cos\big(\omega(k)t\big) - i \, \frac{\sin\big(\omega(k)t\big)}{\omega(k)} H(k) \right) \widehat{\eta}_0(k).
\end{equation}

\section{Derivation of the Hartree equation}\label{sec_proof_Hartree_stuff}

The first step in proving our main theorems is to show that the Hartree equation describes the microscopic dynamics in good approximation as outlined in Section~\ref{sec_main_results_counting}. In order to do so, we prove several preliminary results in this section following the approach from \cite{pickl}. Section~\ref{sec_proof_counting} mostly establishes notation and a straightforward lemma that has been proven before, e.g., in \cite{pickl:2010hartree}. In Section~\ref{section_alpha_der_gen} we establish a result that is similar but more general than what has been proven before, e.g., in \cite{dfpp}. We have included the full proofs anyway to improve the readability of the article.

\subsection{Counting the Number of Particles Outside the Condensate}\label{sec_proof_counting}

Following \cite{pickl} we first define certain weighted operators in terms of the projectors from Definition~\ref{defproa}. We shall also give some general properties of these operators before turning to the special case of the Hartree equation for large volume in Section~\ref{section_alpha_der_gen}.
\begin{definition}\label{defprob}
Let $\phi\in\LZ$ with $\norm{\varphi} = 1$. 
\begin{enumerate}
\item For any $0\leq k\leq N$ we define the set 
\begin{equation}
\mathcal{A}_k:=\Big\{(a_1,a_2,\ldots,a_N): a_k\in\{0,1\}\;;\;
\sum_{j=1}^N a_j=k\Big\}
\end{equation}
and the orthogonal projector $P_{k}^\phi$ acting on $\LZN$ as
\begin{equation}
P_{k}^\phi:= \big( q^{\varphi}_1 \ldots q^{\varphi}_k p^{\varphi}_{k+1} \ldots p^{\varphi}_N \big)_{\sym}:=\sum_{a\in\mathcal{A}_k}\prod_{j=1}^N\big(p_{j}^{\phi}\big)^{1-a_j} \big(q_{j}^{\phi}\big)^{a_j}\;,
\end{equation}
i.e., $P_{k}^\phi$ is the symmetric tensor product of $q^{\varphi}_1, \ldots, q^{\varphi}_k, p^{\varphi}_{k+1}, \ldots, p^{\varphi}_N$. For negative $k$ and $k>N$ we set $P_{k}^\phi:=0$.
\item For any function $f:\{0,1,\ldots,N\}\to\mathbb{R}^+_0$ we define the
operator $\widehat{f}^{\phi}:\LZN\to\LZN$ as
\begin{equation}
\label{hut}\widehat{f}^{\phi}:=\sum_{j=0}^{N} f(j)P_j^\phi\;,
\end{equation}
and, for any $\Psi_N \in L^2(\RRR^{3N})$, the functional
\begin{equation}
\alpha_N(f,\Psi_N,\phi)=\langle\Psi_N,\widehat{f}^\phi\Psi_N\rangle.
\end{equation}
We also define the shifted operators
$\widehat{f}^{\phi}_d:\LZN\to\LZN$ as
\begin{equation}
\widehat{f}^{\phi}_d:=\sum_{j=-d}^{N-d} f(j+d)P_j^\phi\;.
\end{equation}
\end{enumerate}
\end{definition}

\begin{notation}
It should always be clear form the context when hats $\;\widehat{\cdot}\;$ are used in the sense of Definition \ref{defprob} or when they denote Fourier transform. The letter $n$ shall always be used for the function $n(k) = k/N$.
\end{notation}
With Definition~\ref{defproa} and Definition~\ref{defprob} we arrive directly at the following lemma based on combinatorics of the $p_j^\phi$ and $q_j^\phi$. The lemma was already proved (with (c) slightly different), e.g., in \cite{pickl:2010hartree}, and we repeat a proof here for the convenience of the reader.

\begin{lemma}\label{kombinatorik}
Let $\phi\in\LZ$ with $\norm{\varphi} = 1$.
\begin{enumerate}
\item For any functions $l,m:\{0,1,\ldots,N\}\to\mathbb{R}^+_0$ we have
that
\begin{equation}\label{exchange}
\widehat{l}^\phi\widehat{m}^\phi=\widehat{lm}^\phi=\widehat{m}^\phi\widehat{l}^\phi, \qquad\quad \widehat{m}^\phi p_j^\phi=p_j^\phi\widehat{m}^\phi, \qquad\quad \widehat{m}^\phi P_{k}^\phi =P_{k}^\phi\widehat{m}^\phi\;.
\end{equation}
\item Let $n:\{0,1,\ldots,N\}\to\mathbb{R}^+_0$ be given by $n(k)=k/N$.
Then $\widehat{n}^{\phi}$ equals the relative particle number operator of particles not in the state $\phi$, i.e.,
\begin{equation}\label{n_q}
\widehat{n}^{\phi}=N^{-1}\sum_{j=1}^Nq_j^\phi.
\end{equation}
\item For any symmetric $\Psi_N\in\LZN$ we have
\begin{equation}
\norm{q_1^{\varphi}\Psi_N} = \norm{(\widehat{n}^{\varphi})^{1/2}\Psi_N} \quad\text{and}\quad \norm{q_1^{\varphi}q_2^{\varphi}\Psi_N} \leq \norm{\widehat{n}^{\varphi}\Psi_N}.
\end{equation}
\item For any function $m:\{0,1,\ldots,N\}\to\mathbb{R}^+_0$, any function $f:\mathbb{R}^6\to\mathbb{R}$ and any
$j,k=0,1,2$ we have 
\begin{equation}
Q^\phi_k f(x_1,x_2)Q^\phi_j\widehat{m}^\phi = \widehat{m}^\phi_{j-k} Q^\phi_k f(x_1,x_2)Q^\phi_j\;,
\end{equation}
where
$Q^\phi_0:=p^\phi_1 p^\phi_2$, $Q^\phi_1:=p^\phi_1q^\phi_2+q^\phi_1p^\phi_2$ and
$Q^\phi_2:=q^\phi_1q^\phi_2$.
\end{enumerate}
\end{lemma}
\begin{proof}
\begin{enumerate}
\item This follows immediately from Definition~\ref{defproa} and Definition~\ref{defprob}, using that $p_j$ and $q_j$ are orthogonal projectors.
\item  Note that $\cup_{k=0}^N\mathcal{A}_k=\{0,1\}^N$, so $1=\sum_{k=0}^N P_k^\phi$. Using also
$(q_k^\phi)^2=q_k^\phi$ and $q_k^\phi p_k^\phi=0$ we get
\begin{equation}
N^{-1}\sum_{k=1}^Nq_k^\phi = N^{-1}\sum_{k=1}^N q_k^\phi\sum_{j=0}^N P_j^\phi = N^{-1}\sum_{j=0}^N \sum_{k=1}^N q_k^\phi P_j^\phi = N^{-1}\sum_{j=0}^N j P_j^\phi
\end{equation}
and (b) follows.
\item Using symmetry of $\Psi_N$ and \eqref{n_q} we find
\begin{align}
\norm{q_1^{\varphi}\Psi_N}^2 = \scp{\Psi_N}{q_1^{\varphi}\Psi_N} = \scp{\Psi_N}{N^{-1} \sum_{j=1}^N q_j^{\varphi} \Psi_N} = \scp{\Psi_N}{\widehat{n}^{\varphi} \Psi_N} =  \norm{(\widehat{n}^{\varphi})^{1/2}\Psi_N}^2.
\end{align}
Similarly,
\begin{align}
\norm{q_1^{\varphi}q_2^{\varphi}\Psi_N}^2 &= \scp{\Psi_N}{q_1^{\varphi}q_2^{\varphi}\Psi_N} = \scp{\Psi_N}{q_1^{\varphi} (N-1)^{-1} \sum_{j=2}^N q_j^{\varphi} \Psi_N} \nonumber\\
&= \scp{\Psi_N}{q_1^{\varphi} (N-1)^{-1} \sum_{j=1}^N q_j^{\varphi} \Psi_N} - \scp{\Psi_N}{q_1^{\varphi} (N-1)^{-1} \Psi_N} \nonumber\\
&= \scp{\Psi_N}{N^{-1} \sum_{i=1}^N q_i^{\varphi} (N-1)^{-1} \sum_{j=1}^N q_j^{\varphi} \Psi_N} - \scp{\Psi_N}{(N-1)^{-1}N^{-1} \sum_{i=1}^N q_i^{\varphi} \Psi_N} \nonumber\\
&= \scp{\Psi_N}{\sum_{k=1}^N \frac{k^2}{N(N-1)} P_k^{\varphi} \Psi_N} - \scp{\Psi_N}{\sum_{k=1}^N \frac{k}{N(N-1)} P_k^{\varphi} \Psi_N} \nonumber\\
&= \scp{\Psi_N}{\sum_{k=2}^N \frac{k(k-1)}{N(N-1)} P_k^{\varphi} \Psi_N} \leq \scp{\Psi_N}{\sum_{k=2}^N \frac{k^2}{N^2} P_k^{\varphi} \Psi_N} \leq \norm{\widehat{n}^{\varphi}\Psi_N}^2.
\end{align}
\item Using the definitions above we have
\begin{equation}
Q^\phi_k f(x_1,x_2)Q^\phi_j\widehat{m} = \sum_{\ell=0}^N m(\ell) Q^\phi_kf(x_1,x_2)Q^\phi_jP_{\ell}^\phi.
\end{equation}
The number of projectors $q^\phi_j$ in $P^\phi_{\ell} Q^\phi_j$ in the coordinates $j=3,\ldots,N$ is equal to $\ell-j$. The $p^\phi_j$ and $q^\phi_j$ with $j=3,\ldots,N$ commute with $Q^\phi_kf(x_1,x_2)Q^\phi_j$. Thus
$ Q^\phi_kf(x_1,x_2)Q^\phi_jP^\phi_{\ell}= P^\phi_{\ell-j+k}Q^\phi_kf(x_1,x_2)Q^\phi_j$ and
\begin{align}
Q^\phi_k f(x_1,x_2)Q^\phi_j\widehat{m}^\phi &= \sum_{\ell=0}^N  m(\ell) P^\phi_{\ell-j+k}Q^\phi_kf(x_1,x_2)Q^\phi_j \nonumber\\
&= \sum_{\ell=k-j}^{N+k-j} m(\ell+j-k)P^\phi_{\ell}Q^\phi_kf(x_1,x_2) Q^\phi_j = \widehat{m}^\phi_{j-k}Q^\phi_k f(x_1,x_2)Q^\phi_j\;.
\end{align}
\end{enumerate}
\end{proof}

\subsection{Estimating the Number of Particles Outside the Condensate}\label{section_alpha_der_gen}
As presented in \cite{pickl} we wish to control the functional $\alpha_N$
given by 
\begin{equation}\label{alpha_def}
\alpha_N(m,\Psi_N,\phi) := \langle\Psi_N,\widehat{m}^\phi\Psi_N\rangle
\end{equation}
for any weight $m:\{0,\ldots,N\}\to[0,1]$, $\Psi_N\in\LZN$ and $\phi\in\LZ$. We shall need comparatively strong conditions on the ``purity'' of the initial condensate to derive the Hartree equation for large volume. This is encoded in the weights we shall choose below (see Definition \ref{mk}). For these weights convergence of the respective $\alpha$ is stronger than convergence of the reduced density to $|\phi\rangle\langle\phi|$ in operator norm (see \cite{pickl}).

To start with we  give some general statements and estimates for solutions of the Schr\"odinger equation
\begin{equation}\label{Schr_again}
i \partial_t \Psi_N^t = H_N \Psi_N^t, ~\text{with}~ H_N=\sum_{j=1}^N h_j^0 + \rho^{-1}\sum_{j<k}v(x_j-x_k)
\end{equation}
and the respective nonlinear Schr\"odinger equation
\begin{equation}\label{Hartree_again}
i \partial_t\phi_t= h^0\phi_t+ \frac{N-1}{\rho} \big[(v*|\phi_t|^2)-\mu^{\varphi_t}\big] \phi_t.
\end{equation}
For proving Lemma~\ref{cor} we will estimate $\partial_t\alpha_N(m,\Psi_N^t,\phi_t)$ in terms of $\alpha_N(m,\Psi_N^t,\phi_t)$ and a small error, and then use Gronwall's lemma.  To shorten notation we use the following definitions.

\begin{definition}\label{defalpha}
Let
\begin{equation}
W_{j,k}^{\varphi}:=\Lambda(N-1)\bigg(v (x_j-x_k)- (v*|\phi|^2)(x_j) - (v*|\phi|^2)(x_k) + 2\mu^{\varphi}\bigg)\;.
\end{equation}
We define the functionals $\gamma_N^{a,b,c}$ as
\begin{align}
\gamma_N^a(m,\Psi_N,\phi) &:= 2\Im \laa\Psi_N,(\mlf-\mlf_{-1})p^{\varphi}_1q^{\varphi}_2W^{\varphi}_{1,2}p^{\varphi}_1p^{\varphi}_2 \Psi_N\raa, \\
\gamma_N^b(m,\Psi_N,\phi) &:= \Im\laa\Psi_N ,(\mlf-\mlf_{-2})q^{\varphi}_1q^{\varphi}_2W^{\varphi}_{1,2}p^{\varphi}_1p^{\varphi}_2 \Psi_N\raa, \\
\gamma_N^c(m,\Psi_N,\phi) &:= 2\Im\laa\Psi_N,(\mlf-\mlf_{-1})q^{\varphi}_1q^{\varphi}_2W^{\varphi}_{1,2}p^{\varphi}_1q^{\varphi}_2 \Psi_N\raa.
\end{align}
\end{definition} 
The $\gamma_N^{a,b,c}$ are defined in such a way that for any solution of the Schr\"odinger equation $\Psi_N^t$ and any solution $\phi_t$ of the Hartree equation, $\partial_t \alpha_N(m,\Psi_N^t,\phi_t) = \gamma_N^a(m,\Psi_N^t,\phi_t)+\gamma_N^b(m,\Psi_N^t,\phi_t)+\gamma_N^c(m,\Psi_N^t,\phi_t)$, which is shown in Lemma~\ref{ableitung} below. It is left to show that the $\gamma_N^{a,b,c}(m,\Psi_N^t,\phi_t)$ can be controlled by $\alpha_N(m,\Psi_N^t,\phi_t)$ and small error terms (which is done in Lemma~\ref{alphaabl} and Lemma~\ref{cor} below) to close the Gronwall argument. The following lemma is a standard result, see, e.g., \cite{pickl} or \cite{mpp}, and we repeat the proof here for the convenience of the reader.
\begin{lemma}\label{ableitung}
For any solution $\Psi_N^t$ of the Schr\"odinger equation \eqref{Schr_again}, any
solution $\phi_t$ of the Hartree equation \eqref{Hartree_again}, and any weight $m:\{0,1,\ldots,N\}\to\mathbb{R}^+_0$ we have
\begin{equation}\label{firstpart}
\frac{d}{dt} \alpha_N(m,\Psi_N^t,\phi_t) = \gamma_N^a(m,\Psi_N^t,\phi_t) + \gamma_N^b(m,\Psi_N^t,\phi_t) + \gamma_N^c(m,\Psi_N^t,\phi_t)
\end{equation}
with $\gamma_N^a$, $\gamma_N^b$ and $\gamma_N^c$ as in Definition~\ref{defalpha}.
\end{lemma} 
\begin{proof}
Recall the definition
\begin{equation}
H^{\varphi_t}_{\mf}:=\sum_{k=1}^N h^{\varphi_t}_k
\end{equation}
for the sum of Hartree Hamiltonians in each particle. It follows that
\begin{equation}\label{ablp}
\frac{d}{dt}\widehat{m}^{\phi_t}=-i\Big[H^{\varphi_t}_{\mf},\widehat{m}^{\phi_t}\Big]
\end{equation}
for any weight $m:\{0,\ldots,N\}\to\mathbb{R}^+_0$.  With (\ref{ablp}) we get 
\begin{align}
\frac{d}{dt} \alpha _N(m,\Psi_N^t,\phi_t) &= -i\laa\Psi_N^t,\widehat{m}^{\phi_t} H_N\Psi_N^t\raa + i\laa H_N\Psi_N^t,\widehat{m}^{\phi_t} \Psi_N^t\raa - i\laa \Psi_N^t,[H^{\varphi_t}_{\mf},\widehat{m}^{\phi_t}]\Psi_N^t\raa \nonumber\\
&= i\laa\Psi_N^t,[H_N-H^{\varphi_t}_{\mf},\widehat{m}^{\phi_t}]\Psi_N^t\raa\;.
\end{align}
Using symmetry of $\Psi_N^t$ and selfadjointness of $W^{\varphi_t}_{j,k}$ it follows that
\begin{align}\label{wiehier}
\frac{d}{dt} \alpha _N(m,\Psi_N^t,\varphi_t) 
&= i \laa\Psi_N^t,\bigg[\frac{1}{\rho}\sum_{j<k} v(x_j-x_k) - \frac{N-1}{\rho}\sum_{j=1}^N \big(v*|\phi_t|^2\big)(x_j) , \widehat{m}^{\varphi_t} \bigg]\Psi_N^t\raa \nonumber\\
&= \frac{i}{2}\laa\Psi_N^t,\big[W^{\varphi_t}_{1,2},\widehat{m}^{\varphi_t} \big]\Psi_N^t\raa.
\end{align}
Let us next establish a formula for the commutator.
For any function $f:\RRR^6\to\RRR$, any $\phi\in L^2(\RRR^3)$ and any weight $m:\{0,\ldots,N\}\to\mathbb{R}^+_0$ we have that, using again the notation $Q^\phi_0:=p^\phi_1 p^\phi_2$, $Q^\phi_1:=p^\phi_1q^\phi_2+q^\phi_1p^\phi_2$ and
$Q^\phi_2:=q^\phi_1q^\phi_2$ (note that $\sum_{k=0}^2Q^\phi_k=1$) from above, and Lemma~\ref{kombinatorik} (d),
\begin{align}
\big[f(x_1,x_2),\widehat{m}^\phi\big] &= \sum_{k,j=0}^2 \bigg( Q^\phi_kf(x_1,x_2) Q^\phi_j\widehat{m}^\phi - \widehat{m}^\phi Q^\phi_k f(x_1,x_2)Q^\phi_j\bigg) \nonumber\\
&= \sum_{k>j}(\widehat{m}^\phi_{j-k}-\widehat{m}^\phi)Q^\phi_k f(x_1,x_2) Q^\phi_j + \sum_{k<j}Q^\phi_k f(x_1,x_2)Q^\phi_j (\widehat{m}^\phi-\widehat{m}^\phi_{k-j}) \nonumber\\
&= \sum_{k>j}(\widehat{m}^\phi_{j-k}-\widehat{m}^\phi)Q^\phi_k f(x_1,x_2) Q^\phi_j - \hc,
\end{align}
where $\hc$ stands for the hermitian conjugate of the preceding term. Thus, for $f(x_1,x_2)=W^{\varphi_t}_{1,2}$, we get from \eqref{wiehier}
\begin{equation}
\frac{d}{dt} \alpha _N(m,\Psi_N^t,\varphi_t) = \sum_{k>j}\Im \laa\Psi_N^t,(\widehat{m}^{\varphi_t}-\widehat{m}^{\varphi_t}_{j-k})Q^{\varphi_t}_kW^{\varphi_t}_{1,2}Q^{\varphi_t}_j\Psi_N^t\raa.
\end{equation}
Now \eqref{firstpart} follows using the symmetry of $W^{\varphi_t}_{1,2}$ and $\Psi_N^t$, i.e., $Q^{\varphi_t}_1=p^{\varphi_t}_1q^{\varphi_t}_2+q^{\varphi_t}_1p^{\varphi_t}_2$ can be replaced by $2p^{\varphi_t}_1q^{\varphi_t}_2$ in the scalar product.
\end{proof}
With Lemma~\ref{ableitung}, a good control of $\alpha_N$ follows once we can control the different summands appearing in the $\gamma_N$'s in a suitable way. Parts of the following lemma have been proven in a slightly less general form in different steps of the proof of the main result in \cite{dfpp}.
\begin{lemma}\label{alphaabl}
Let $m:\{0,\ldots,N\}\to\mathbb{R}^+_0$ be monotone increasing and let $l:\{0,\ldots,N\}\to\mathbb{R}_0^+$ be such that for any $k\in\mathbb{Z}$
\begin{align}
m(k+2)-m(k) &\leq l(k),~~~~~~ m(k)-m(k-2) \leq l(k), \nonumber\\
m(k+1)-m(k) &\leq l(k),~~~~~~ m(k)-m(k-1) \leq l(k).
\end{align}
Then, for any $\phi \in L^\infty(\RRR^3)$ and any $\Psi_N\in L^2(\RRR^{3N})$, we have
\begin{align}
\gamma_N^a(m,\Psi_N,\phi) &= 0 \label{gamma_a}, \\
\gamma_N^b(m,\Psi_N,\phi) &\leq 2\Lambda N\|\phi\|_\infty^2 \big(\|v\|_1+\|v\|_2\big)\laa\Psi_N,\widehat{l}^{\varphi}  \widehat{n}^{\varphi}\Psi_N\raa + \|v\|_2\Lambda \, \laa\Psi_N, \widehat{l}^{\varphi} \Psi_N\raa \label{gamma_b}, \\
\gamma_N^c(m,\Psi_N,\phi) &\leq 6\Lambda N \|\phi\|_\infty \|v\|_2 \left(\laa\Psi_N,\widehat{l}^{\varphi} (\widehat{n}^{\varphi})^2\Psi_N\raa\; \laa\Psi_N,\widehat{l}^{\varphi} \widehat{n}^{\varphi} \Psi_N\raa\right)^{1/2}. \label{gamma_c}
\end{align}
\end{lemma}
Before we prove the lemma, let us briefly discuss how to estimate the $\gamma_N^{a,b,c}$. It is $\gamma_N^{a}$ which is physically the most important. Here the mean field cancels out most of the interaction. The central point in the mean-field argument is to observe that $p^{\phi}_1q^{\phi}_2W^{\varphi}_{1,2}p^{\phi}_1p^{\phi}_2$ is small or even zero in our case. For $\gamma_N^b$ we use the symmetrization introduced in \cite{pickl} to bring one of the two operators $q^{\phi}$ to the other side of the interaction term, in order to gain the expectation value of $\widehat{l}^{\varphi}\,\widehat{n}^{\varphi}$ which for suitable weight functions can be estimated in terms of $\alpha_N$. Note that the symmetrization works due to the symmetry of the wave function and gives a small error in terms of the expectation value of $\widehat{l}^{\varphi}$. For $\gamma_N^c$ the choice of the weight $m$ plays an important role. Note that we only have one projector $p^{\phi}$ here and $\|q^{\phi}_1q^{\phi}_2W^{\varphi}_{1,2} p^{\phi}_1q^{\phi}_2\|_{op}$ cannot just be controlled by the $L^1$-norm of $v$ times $\Lambda^{-1}$. On the other hand, we have altogether three projectors $q^{\phi}$ in $\gamma_N^c$. Assuming that the condensate is very clean (which is encoded in $\widehat{m}^{\phi}$), these $q^{\phi}$'s make $\gamma_N^c$ small.

\begin{proof}[Proof of Lemma~\ref{alphaabl}]
Recall the definition
\begin{equation}
W^{\varphi}_{1,2}:=\Lambda(N-1)\bigg(v(x_1-x_2)- \big(v*|\phi|^2\big)(x_1)-\big(v*|\phi|^2\big)(x_2) + 2\mu^{\varphi}\bigg).
\end{equation}

\textbf{Estimate for $\gamma_N^a$.} Using the bra-ket notation $p_1^\phi=|\phi(x_1)\rangle\langle\phi(x_1)|$ we find
\begin{align}\label{projonf}
p_1^\phi v(x_1-x_2)p_1^\phi = |\phi(x_1)\rangle\langle\phi(x_1)|v(x_1-x_2)|\phi(x_1)\rangle\langle\phi(x_1)| = p_1^\phi (v*|\phi|^2)(x_2),
\end{align}
so from the definition of $W^{\varphi}_{1,2}$ and from $p_2^\phi q_2^\phi=0$ it follows that $p_1^{\phi}q_2^{\phi}W^{\varphi}_{1,2}p_1^{\phi}p_2^{\phi} = 0$ and thus that $\gamma_N^a=0$.

\textbf{Estimate for $\gamma_N^b$.} We use first that $q^{\varphi}_1q^{\varphi}_2 f(x_1)p^{\varphi}_1p^{\varphi}_2=0$ for any $f:\RRR^3\to\RRR^3$, so
\begin{equation}
\gamma_N^b(m,\Psi_N,\phi) = \Lambda (N-1)\, \Im \laa\Psi_N,(\mlf - \mlf_{-2})q_1^\phi q_2^\phi v(x_1-x_2) p_1^\phi p_2^\phi \Psi_N\raa.
\end{equation}
Let us introduce the short-hand notation
\begin{align}\label{discrete_derivatives}
\widehat{m^{'(d)}}^{\phi} &:= \sum_{k=d}^N \big(m(k) - m(k-d)\Big) P_k^{\varphi}, \nonumber\\
\widehat{m^{'(-d)}}^{\phi} &:= \sum_{k=0}^{N-d} \big(m(k+d) - m(k)\Big) P_k^{\varphi} = \widehat{m^{'(d)}}^{\phi}_d,
\end{align}
such that $(\mlf - \mlf_{-2})q_1^\phi q_2^\phi = \widehat{m^{'(2)}}^{\phi} q_1^\phi q_2^\phi$ (note that without the $q$ projectors there would be extra $k=0$ and $k=1$ terms). By monotonicity of $m$ we find $\widehat{m^{'(2)}}^{\phi} \geq 0$ so with Lemma~\ref{kombinatorik} (d) we can write
\begin{equation}\label{suffsym}
\gamma_N^b(m,\Psi_N,\phi) = \Lambda (N-1)\, \Im \laa \Psi_N,\Big(\widehat{m^{'(2)}}^{\phi}\Big)^{1/2}q_1^\phi q_2^\phi v(x_1-x_2) p_1^\phi p_2^\phi\Big(\widehat{m^{'(-2)}}^{\phi}\Big)^{1/2} \Psi_N\raa.
\end{equation}
Before we estimate this term note that the operator norm of $q^{\varphi}_1q^{\varphi}_2v(x_1-x_2)$ restricted to the subspace of symmetric functions is much smaller than the operator norm on all of $\LZN$. This comes from the fact that $v(x_1-x_2)$ is only nonzero in a small domain where $x_1\approx x_2$. A non-symmetric wave function may be fully localized in that area, whereas for a symmetric wave function only a small part lies in that area. To get sufficiently good control of \eqref{suffsym} we ``symmetrize'' $(N-1)v(x_1-x_2)$, replacing it by
$\sum_{j=2}^Nv(x_1-x_j)$. This leads to
\begin{align}\label{product}
\gamma_N^b(m,\Psi_N,\phi) &= \Lambda(N-1)\big\laa\Psi_N,\Big(\widehat{m^{'(2)}}^{\phi}\Big)^{1/2}q^{\varphi}_1q^{\varphi}_2 v(x_1-x_2) p^{\varphi}_1p^{\varphi}_2 \Big(\widehat{m^{'(-2)}}^{\phi}\Big)^{1/2} \Psi_N\big\raa \nonumber \\
&= \Lambda\big\laa\Psi_N ,\Big(\widehat{m^{'(2)}}^{\phi}\Big)^{1/2}\sum_{j=2}^Nq^{\varphi}_1q^{\varphi}_jv(x_1-x_j)p^{\varphi}_1p^{\varphi}_j \Big(\widehat{m^{'(-2)}}^{\phi}\Big)^{1/2}\Psi_N\big\raa \nonumber\\
&\leq \Lambda\Big\|q^{\varphi}_1\Big(\widehat{m^{'(2)}}^{\phi}\Big)^{1/2}\Psi_N\Big\|\;\Bigg\|\sum_{j=2}^Nq^{\varphi}_jv(x_1-x_j)p^{\varphi}_1p^{\varphi}_j \Big(\widehat{m^{'(-2)}}^{\phi}\Big)^{1/2}\Psi_N\Bigg\|\;.
\end{align}
For the first factor we have
in view of Lemma \ref{kombinatorik} (c) that 
\begin{equation}\label{factor}
\|q^{\varphi}_1\Big(\widehat{m^{'(2)}}^{\phi}\Big)^{1/2}\Psi_N\|^2\;=\;\laa\Psi_N\Big(\widehat{m^{'(2)}}^{\phi}\Big)\widehat{n}^{\varphi}\Psi_N\raa\;.
\end{equation}
For the second factor of \eqref{product} we calculate \begin{align}\label{zweifreunde}
&\Bigg\|\sum_{j=2}^Nq^{\varphi}_jv(x_1-x_j)p^{\varphi}_1p^{\varphi}_j \Big(\widehat{m^{'(-2)}}^{\phi}\Big)^{1/2}\Psi_N\Bigg\|^2 \nonumber\\
&\qquad = \nonumber\sum_{\substack{j, k = 2 \\ j\neq k}}^N\laa\Big(\widehat{m^{'(-2)}}^{\phi}\Big)^{1/2}\Psi_N, p^{\varphi}_1p^{\varphi}_j
v(x_1-x_j)q^{\varphi}_j q^{\varphi}_kv(x_1-x_k)p^{\varphi}_1p^{\varphi}_k \Big(\widehat{m^{'(-2)}}^{\phi}\Big)^{1/2} \Psi_N\raa\nonumber\\
&\qquad\quad + \sum_{k=2}^N\|q^{\varphi}_kv(x_1-x_k)p^{\varphi}_1p^{\varphi}_k\Big(\widehat{m^{'(-2)}}^{\phi}\Big)^{1/2}\Psi_N\|^2\;.
\end{align}
Using symmetry the first summand in \eqref{zweifreunde} is bounded by
\begin{align}
&N^2 \big| \laa\Big(\widehat{m^{'(-2)}}^{\phi}\Big)^{1/2}\Psi_N, p^{\varphi}_1p^{\varphi}_2q^{\varphi}_3v(x_1-x_2)v(x_1-x_3)p^{\varphi}_1q^{\varphi}_2p^{\varphi}_3\Big(\widehat{m^{'(-2)}}^{\phi}\Big)^{1/2}\Psi_N\raa \big| \nonumber\\
&\qquad\leq N^2\|\sqrt{|v(x_1-x_2)|}\sqrt{|v(x_1-x_3)|}p^{\varphi}_1q^{\varphi}_2p^{\varphi}_3\Big(\widehat{m^{'(-2)}}^{\phi}\Big)^{1/2}\Psi_N\|^2 \nonumber\\
&\qquad\leq N^2\|\sqrt{|v(x_1-x_2)|}p^{\varphi}_1\|^4_{op}\;\|q^{\varphi}_2\Big(\widehat{m^{'(-2)}}^{\phi}\Big)^{1/2}\Psi_N\|^2 \nonumber\\
&\qquad\leq N^2\|\varphi\|_\infty^4 \|v\|_1^2\|q^{\varphi}_2\Big(\widehat{m^{'(-2)}}^{\phi}\Big)^{1/2}\Psi_N\|^2 \nonumber\\
&\qquad\leq N^2\|\varphi\|_\infty^4 \|v\|_1^2\laa \Psi_N,\Big(\widehat{m^{'(-2)}}^{\phi}\Big)\widehat{n}^{\varphi}\Psi_N\raa \;.
\end{align}
The second summand in \eqref{zweifreunde} can be controlled by
\begin{align}
&N\laa\Big(\widehat{m^{'(-2)}}^{\phi}\Big)^{1/2}\Psi_N, p^{\varphi}_1p^{\varphi}_2(v(x_1-x_2))^2p^{\varphi}_1p^{\varphi}_2\Big(\widehat{m^{'(-2)}}^{\phi}\Big)^{1/2}\Psi_N\raa \nonumber\\
&\qquad\leq N \|p^{\varphi}_1(v(x_1-x_2))^2p^{\varphi}_1 \|_{op}\;\|\Big(\widehat{m^{'(-2)}}^{\phi}\Big)^{1/2}\Psi_N\|^2 \nonumber\\
&\qquad\leq N\|v\|_2^2\|\phi\|^2_\infty\; \laa\Psi_N \Big(\widehat{m^{'(-2)}}^{\phi}\Big)\Psi_N\raa \;.
\end{align}
Using $\sqrt{a+b} \leq \sqrt{a} + \sqrt{b}$ it follows that
\begin{equation}
\sqrt{\eqref{zweifreunde}}\leq N\|\varphi\|_\infty^2 \|v\|_1\laa \Psi_N,\Big(\widehat{m^{'(-2)}}^{\phi}\Big)\widehat{n}^{\varphi}\Psi_N\raa^{1/2} + \sqrt{N}\|v\|_2\|\varphi\|_\infty\; \laa\Psi_N \Big(\widehat{m^{'(-2)}}^{\phi}\Big)\Psi_N\raa^{1/2} \;.
\end{equation}
Since $\gamma_N^b(m,\Psi_N,\phi) \leq \Lambda\sqrt{\eqref{factor}\eqref{zweifreunde}}$ (see \eqref{product}) we get that
\begin{align}
\gamma_N^b(m,\Psi_N,\phi) &\leq \Lambda \laa\Psi_N\Big(\widehat{m^{'(2)}}^{\phi}\Big)\widehat{n}^{\varphi}\Psi_N\raa^{1/2} N\|\varphi\|_\infty^2 \|v\|_1\laa \Psi_N,\Big(\widehat{m^{'(-2)}}^{\phi}\Big)\widehat{n}^{\varphi}\Psi_N\raa^{1/2} \nonumber\\
&\quad + \Lambda \laa\Psi_N\Big(\widehat{m^{'(2)}}^{\phi}\Big)\widehat{n}^{\varphi}\Psi_N\raa^{1/2}\sqrt{N}\|v\|_2\|\varphi\|_\infty\; \laa\Psi_N \Big(\widehat{m^{'(-2)}}^{\phi}\Big) \Psi_N\raa^{1/2}.
\end{align}
Using $ab \leq a^2/2 + b^2/2$ and the assumptions on $m$ that $m(k+2)-m(k)\leq l(k)$ and $m(k)-m(k-2) \leq l(k)$, we get that
\begin{align}
\gamma_N^b(m,\Psi_N,\phi) \leq \Lambda N\|\varphi\|_\infty^2
(\|v\|_1+\|v\|_2)\laa\Psi_N\widehat{l}^{\varphi}  \widehat{n}^{\varphi}\Psi_N\raa + \|v\|_2\Lambda \laa\Psi_N \widehat{l}^{\varphi} \Psi_N\raa.
\end{align}

\textbf{Estimate for $\gamma_N^c$.} Let us introduce
\begin{equation}
\widetilde{W}_{j,k}^{\varphi}:=\Lambda(N-1)\bigg(v (x_j-x_k)- (v*|\phi|^2)(x_j) - (v*|\phi|^2)(x_k)\bigg)\;,
\end{equation}
i.e., $W_{j,k}^{\varphi}$ without the $\mu^{\varphi}$ term. Using Lemma \ref{kombinatorik} (d) and Cauchy-Schwarz we get
\begin{align}
\gamma_N^c &= 2 \Im \laa\Psi_N ,\Big(\widehat{m^{'(1)}}^{\phi}\Big)^{1/2}q^{\varphi}_1q^{\varphi}_2\widetilde{W}^{\varphi}_{1,2} p^{\varphi}_1q^{\varphi}_2 \Big(\widehat{m^{'(-1)}}^{\phi}\Big)^{1/2} \Psi_N\raa \nonumber\\
&\leq 2 \| q^{\varphi}_1q^{\varphi}_2 \Big(\widehat{m^{'(1)}}^{\phi}\Big)^{1/2} \Psi_N\|\; \|\widetilde{W}^{\varphi}_{1,2} p^{\varphi}_1q^{\varphi}_2 \Big(\widehat{m^{'(-1)}}^{\phi}\Big)^{1/2} \Psi_N\|\;.
\end{align}
For both factors we use Lemma \ref{kombinatorik} (c) to get the bound
\begin{equation}
\gamma_N^c \leq 2 \|\Big(\widehat{m^{'(1)}}^{\phi}\Big)^{1/2} \widehat{n}^{\varphi}\Psi_N\|\; \|\widetilde{W}^{\varphi}_{1,2}p^{\varphi}_1\|_{op} \|\Big(\widehat{m^{'(-1)}}^{\phi}\Big)^{1/2}(\widehat{n}^{\varphi})^{1/2} \Psi_N\|\;.
\end{equation}
Note that by using the triangle inequality,
\begin{equation}
\|\widetilde{W}^{\varphi}_{1,2}p^{\varphi}_1\|_{op}\leq \Lambda N\|v(x_1-x_2)p^{\varphi}_1\|_{op} + 2 \Lambda N \|(v*|\varphi|^2)p^{\varphi}_1\|_{op}\leq 3\Lambda N \|\varphi\|_\infty \|v\|_2\;.
\end{equation}
Thus, using the assumptions on $m$, we have
\begin{equation}
\gamma_N^c \leq 6 \Lambda N \|\varphi\|_\infty \|v\|_2 \left(\laa\Psi_N,\widehat{l}^{\varphi} (\widehat{n}^{\varphi})^2\Psi_N\raa\;
\laa\Psi_N,\widehat{l}^{\varphi} \widehat{n}^{\varphi} \Psi_N\raa\right)^{1/2}.
\end{equation}
\end{proof}

\section{Beyond the Mean-field Description}\label{sec_proof_the_beyond}

In order to prove convergence to the Bogoliubov time evolution, it is useful to introduce an auxiliary wave function $\widetilde{\Psi}_N^t$. We define it as a solution to the Schr\"odinger equation
\begin{equation}\label{psitilde}
i\partial_t\widetilde{\Psi}_N^t=\widetilde{H}_N^{\varphi_t}\widetilde{\Psi}_N^t
\end{equation}
with the Hamiltonian
\begin{align}\label{tilde_H}
\widetilde{H}_N^{\varphi_t} &= H^{\varphi_t}_{\mf} + \frac{1}{\rho} \sum_{i<j} \bigg( p^{\varphi_t}_i q^{\varphi_t}_j v(x_i-x_j) q^{\varphi_t}_i p^{\varphi_t}_j + q^{\varphi_t}_i p^{\varphi_t}_j v(x_i-x_j) p^{\varphi_t}_i q^{\varphi_t}_j \nonumber\\
&\qquad\qquad\qquad\quad + p^{\varphi_t}_i p^{\varphi_t}_j v(x_i-x_j) q^{\varphi_t}_i q^{\varphi_t}_j + q^{\varphi_t}_i q^{\varphi_t}_j v(x_i-x_j) p^{\varphi_t}_i p^{\varphi_t}_j \bigg),
\end{align}
where
\begin{equation}
H^{\varphi_t}_{\mf}:=\sum_{k=1}^N h^{\varphi_t}_k.
\end{equation}
Note that for the $\varphi_t$'s which we consider in our main theorems, $\widetilde{H}_N^{\varphi_t}$ generates a unitary propagator $\widetilde{U}(t,s)$. This can be shown, e.g., by using the Lipschitz continuity of the time-dependent part in \eqref{tilde_H} and applying \cite[Theorem~2.5]{mcgriesemer:2017}. Assuming that $\widetilde \Psi_N^0$ has only pair correlations, it follows that $\widetilde{\Psi}_N^t$---although its time evolution was subject to some interaction---does not have any correlations higher than two. This is a direct consequence of the definition of the Hamiltonian $\widetilde{H}_N^{\varphi_t}$. None of the terms in this Hamiltonian can lead to a correlation of three or more pairs. For example the term $q^{\varphi_t}_2q^{\varphi_t}_3v(x_2-x_3)p^{\varphi_t}_2p^{\varphi_t}_3$ acting on a wave function where particle one and two and particle three and four are correlated gives a state where particle two and three are in the state $\phi_t$ while one and four are correlated. We will later prove Theorem~\ref{main_Bog_thm} in two steps: by first proving convergence of $\Psi_N^t$ to $\widetilde{\Psi}_N^t$, and then convergence of $\widetilde{\Psi}_N^t$ to the Bogoliubov time evolution.

\subsection{Choice of a Suitable Weight Function}\label{sec_proof_cor}

In this subsection, we prove the bounds presented in Section~\ref{sec_main_results_counting}. While the bounds for $\widetilde{\Psi}_N^t$ can be established relatively straightforward, the bounds for $\Psi_N^t$ are much harder to prove. These are indeed the main technical novelty of this article. The key to proving bounds on the number of excitations around the mean-field for $\Psi_N^t$ is the choice of the following weight function.
\begin{definition}\label{mk}
We define the function $w:\{0,\ldots,N\}\to\mathbb{R}^+_0$ as
\begin{equation}
w(k):=\left\{ \begin{array}{ll} (k+1)/\rho, & \hbox{for $k+1\leq \rho$,} \\
1, & \hbox{else.} \end{array} \right.
\end{equation}
\end{definition}
We use powers of $\omega$ as a weight function in order to control the $\gamma_N^c$ term from Lemma~\ref{alphaabl}. Using $(k+1)/\rho$ instead of $k/\rho$ makes some of the notation shorter, but we could just as well choose the latter. With this definition we get the following result. 
\begin{lemma}\label{cor} 
Let $\Lambda \geq 1$, $j,\ell \in \NNN$, and assume that
\begin{equation}\label{constant} 
D_i:= 3^i \sup_{0\leq t\leq T} 36 (1+\sqrt\Lambda \|\phi_t\|_\infty)^2 (\|v\|_1+\|v\|_2+\|v\|_\infty)
\end{equation}
is finite. Let $\Psi_N^0$, $\widetilde{\Psi}_N^0$ and $\varphi_0$ be normalized, and $\Psi_N^t$, $\widetilde{\Psi}_N^t$ and $\varphi_t$ be the solutions to \eqref{schroe}, \eqref{psitilde} and \eqref{meanfield}, respectively.
\begin{enumerate}
\item[(a)] Assume
\begin{align}\label{cor_eq_ini_cond}
\alpha_N(n^i,\Psi_N^0,\phi_0) \leq c_i \rho^{-i}
\end{align}
for all $i \leq \max(j,\ell)$, and set $C_i = 1 + \sum_{m=1}^i (2^m c_m + 1)$. Then
\begin{equation}\label{coreins}
\alpha_N(w^{i},\Psi_N^t,\phi_t) \leq C_i \,e^{D_it} \left(\frac{\Lambda}{\rho}\right)^{i},
\end{equation}
and
\begin{equation}\label{corzwei}
\alpha_N(n^{\ell},\Psi_N^t,\phi_t) \leq C_{\ell} e^{D_{\ell}t} \rho^{-\ell} + C_j e^{D_jt} \left(\frac{\Lambda}{\rho}\right)^j.
\end{equation}
\item[(b)] Assume
\begin{align}\label{cor_eq_ini_cond2}
\alpha_N(n^i,\widetilde{\Psi}_N^0,\phi_0) \leq \widetilde{c}_i \rho^{-i}
\end{align}
for all $i \leq \ell$, and set $\widetilde{C}_{\ell} = 1 + \sum_{m=1}^{\ell} (2^m \widetilde{c}_m + 1)$. Then
\begin{equation}\label{cordrei}
\alpha_N(n^{\ell},\widetilde\Psi_N^t,\phi_t) \leq \widetilde{C}_{\ell} e^{D_{\ell}t} \rho^{-\ell}.
\end{equation}
\end{enumerate}
\end{lemma}

\begin{proof}
\begin{enumerate}
\item[(a)] Choosing for any $i \geq 1$ the weight $l_i$ given by
\begin{equation}\label{l_j_def}
l_i(k) := \left\{
\begin{array}{ll}
   \frac{3^i}{\rho} w^{i-1}(k), & \hbox{for $k\leq \rho+1$,} \\
   0, & \hbox{else,}
\end{array}\right.
\end{equation}
it follows that
\begin{align}
w^i(k+2)-w^i(k)\leq l_i(k), \hspace{2cm} w^i(k)-w^i(k-2)\leq l_i(k), \nonumber\\
w^i(k+1)-w^i(k)\leq l_i(k), \hspace{2cm} w^i(k)-w^i(k-1)\leq l_i(k)\;.
\end{align}
Thus, we find
\begin{equation}
\frac{k}{N}l_i(k)\leq 2\frac{3^i}{N} w^{i}(k),\hspace{1cm}\frac{k^2}{N^2}l_i(k)\leq 4\frac{3^i}{N\Lambda} w^{i}(k)\;.
\end{equation}
Note that the extra $\Lambda^{-1}$ in the second estimate is gained by choosing the weight function $\omega(k)$ appropriately. This is the crucial point that allows us to control the $\gamma_N^c$ term. We then get with the estimates from Lemma~\ref{alphaabl} in \eqref{firstpart} that
\begin{align}
\frac{d}{dt}\alpha_N(w^{i},\Psi_N^t,\phi_t) &\leq 2\Lambda 
N\|\phi_t\|_\infty^2
(\|v\|_1+\|v\|_2)\laa\Psi_N^t,\widehat{l}_i^{\varphi_t}  \widehat{n}^{\varphi_t}\Psi_N^t\raa
+\|v\|_2\Lambda
\laa\Psi_N^t, \widehat{l}_i^{\varphi_t} \Psi_N^t\raa
\nonumber\\
&\quad +6\Lambda N \|\phi_t\|_\infty \|v\|_2 \left(\laa\Psi_N^t,\widehat{l}_i^{\varphi_t} (\widehat{n}^{\varphi_t})^2\Psi_N^t\raa\;
\laa\Psi_N^t,\widehat{l}^{\varphi_t}_i \widehat{n}^{\varphi_t} \Psi_N^t\raa\right)^{1/2}
\nonumber\\
&\leq 4\,3^i \Lambda 
\|\phi_t\|_\infty^2
(\|v\|_1+\|v\|_2)\laa\Psi_N^t, (\widehat{w}^{\varphi_t})^{i}\Psi_N^t\raa
\nonumber\\
&\quad + 3^i \|v\|_2\frac{\Lambda}{\rho}
\laa\Psi_N^t, (\widehat{w}^{\varphi_t})^{i-1} \Psi_N^t\raa + 18\, 3^i  \Lambda^{1/2}  \|\phi_t\|_\infty \|v\|_2  \laa\Psi_N^t,(\widehat{w}^{\varphi_t})^{i} \Psi_N^t\raa
\nonumber\\
&\leq 18 \, 3^i (1+\sqrt{\Lambda} 
\|\phi_t\|_\infty )^2
(\|v\|_1+\|v\|_2)\laa\Psi_N^t, (\widehat{w}^{\varphi_t})^{i}\Psi_N^t\raa
\nonumber\\&\quad+ 3^i \|v\|_2\frac{\Lambda}{\rho}
\laa\Psi_N^t, (\widehat{w}^{\varphi_t})^{i-1} \Psi_N^t\raa\;.
\end{align}
In terms of the constant \eqref{constant} we get
\begin{equation}\label{alpha_der_psi}
\frac{d}{dt}\alpha_N(w^{i},\Psi_N^t,\phi_t) \leq \frac{1}{2}D_i \left(\alpha_N(w^{i},\Psi_N^t,\phi_t)
+\frac{\Lambda}{\rho}\alpha_N(w^{i-1},\Psi_N^t,\phi_t)\right) \;.
\end{equation}
Let us now define
\begin{equation}
\beta_m(t) := \sum_{i=1}^m \left(\frac{\Lambda}{\rho}\right)^{m-i} \alpha_N(w^{i},\Psi_N^t,\phi_t).
\end{equation}
Then \eqref{alpha_der_psi} gives
\begin{align}
\frac{d}{dt} \beta_m(t) &\leq \sum_{i=1}^m \frac{1}{2} D_i \left(\frac{\Lambda}{\rho}\right)^{m-i} \alpha_N(w^{i},\Psi_N^t,\phi_t) + \sum_{i=0}^{m-1} \frac{1}{2} D_{i+1} \left(\frac{\Lambda}{\rho}\right)^{k-i} \alpha_N(w^{i},\Psi_N^t,\phi_t) \nonumber\\
&\leq D_m \left(\beta_m(t) + \left(\frac{\Lambda}{\rho}\right)^m \right),
\end{align}
and thus with Gronwall's lemma,
\begin{align}
\alpha_N(w^{m},\Psi_N^t,\phi_t) \leq \beta_m(t) &\leq e^{D_mt} \sum_{i=0}^m \left(\frac{\Lambda}{\rho}\right)^{m-i} \alpha_N(w^{i},\Psi_N^0,\phi_0).
\end{align}
Finally, note that $\alpha_N(w^{i},\Psi_N^0,\phi_0) \leq 2^i\Lambda^i \alpha_N(n^{i},\Psi_N^0,\phi_0) + \rho^{-i}$, so with our assumption for the initial conditions in \eqref{cor_eq_ini_cond} (and $\Lambda \geq 1$) we have proven \eqref{coreins}. To get \eqref{corzwei} note that for $0 \leq k\leq \rho - 1$ and any $\ell\in\mathbb{N}$,
\begin{equation}
\left(\frac{k}{N}\right)^{\ell}\leq \Lambda^{-\ell}w^{\ell}(k),
\end{equation}
and for $N \geq k > \rho-1$ and any $j,\ell\in\mathbb{N}$,
\begin{equation}
\left(\frac{k}{N}\right)^{\ell}\leq 1 = w^j(k).
\end{equation}
Thus, for any $0\leq k \leq N$, and any $j,\ell\in\NNN$,
\begin{equation}
\left(\frac{k}{N}\right)^{\ell}\leq \Lambda^{-\ell}w^{\ell}(k)+w^j(k)\;,
\end{equation}
i.e., \eqref{coreins} implies \eqref{corzwei}.

\item[(b)] The estimate for $\widetilde{\Psi}_N^t$ goes very similar. It is simpler since $\widetilde{\Psi}_N^t$ is constructed in such a way that the  $\gamma_N^c$ term in Lemma~\ref{ableitung} disappears. Therefore, we can choose powers of $(k+1)/N$ as the weight function $m$. With the weights
\begin{equation}\label{l_j_def_b}
\widetilde{n}(k) := \left\{
\begin{array}{ll}
   \frac{k+1}{N}, & \hbox{for $k\leq N-1$,} \\
   1, & \hbox{for $k = N$},
\end{array}\right. \qquad
\widetilde{l}_i(k) := \frac{3^i}{N} \widetilde{n}^{i-1}(k),
\end{equation}
it follows that
\begin{align}
\widetilde{n}^i(k+2)-\widetilde{n}^i(k)\leq \widetilde{l}_i(k), \hspace{2cm} \widetilde{n}^i(k)-\widetilde{n}^i(k-2)\leq \widetilde{l}_i(k), \nonumber\\
\widetilde{n}^i(k+1)-\widetilde{n}^i(k)\leq \widetilde{l}_i(k), \hspace{2cm} \widetilde{n}^i(k)-\widetilde{n}^i(k-1)\leq \widetilde{l}_i(k)\;,
\end{align}
and
\begin{equation}
\frac{k}{N}\widetilde{l}_i(k)\leq 2\frac{3^i}{N} \widetilde{n}^{i}(k).
\end{equation}
So from \eqref{firstpart} with $\gamma_N^c=0$ and Lemma~\ref{alphaabl}, we find that
\begin{align}
\frac{d}{dt}\alpha_N(\widetilde{n}^{i},\widetilde{\Psi}_N^t,\phi_t) &\leq 2\Lambda 
N\|\phi_t\|_\infty^2
(\|v\|_1+\|v\|_2)\laa\widetilde{\Psi}_N^t,\widehat{\widetilde{l}}^{\,\varphi_t}_i  \widehat{n}^{\varphi_t}\widetilde{\Psi}_N^t\raa
+ \|v\|_2\Lambda
\laa\widetilde{\Psi}_N^t, \widehat{\widetilde{l}}^{\,\varphi_t}_i \widetilde{\Psi}_N^t\raa
\nonumber\\
&\leq 4\,3^i \Lambda 
\|\phi_t\|_\infty^2
(\|v\|_1+\|v\|_2)\laa\widetilde{\Psi}_N^t, \Big(\widehat{\widetilde{n}}^{\varphi_t}\Big)^{i}\widetilde{\Psi}_N^t\raa \nonumber\\
&\quad + 3^i \|v\|_2 \, \rho^{-1} \,
\laa\widetilde{\Psi}_N^t, \Big(\widehat{\widetilde{n}}^{\varphi_t}\Big)^{i-1} \widetilde{\Psi}_N^t\raa
\nonumber\\
&\leq \frac{1}{2}D_i \left(\alpha_N(\widetilde{n}^{i},\widetilde{\Psi}_N^t,\phi_t)
+\rho^{-1}\alpha_N(\widetilde{n}^{i-1},\widetilde{\Psi}_N^t,\phi_t)\right) \;.
\end{align}
We then define
\begin{equation}
\widetilde{\beta}_m(t) := \sum_{i=1}^m \rho^{-m+i} \alpha_N(\widetilde{n}^{i},\widetilde{\Psi}_N^t,\phi_t),
\end{equation}
such that similarly to part (a),
\begin{align}
\frac{d}{dt} \widetilde{\beta}_m(t) &\leq D_m \left(\widetilde{\beta}_m(t) + \rho^{-m} \right).
\end{align}
Finally, Gronwall's lemma and $\alpha_N(\widetilde{n}^{i},\widetilde{\Psi}_N^0,\phi_0) \leq 2^i \alpha_N(n^{i},\widetilde{\Psi}_N^0,\phi_0) + N^{-i}$ with the assumption on the initial conditions \eqref{cor_eq_ini_cond2} (and $\Lambda \geq 1$) yield \eqref{cordrei}.
\end{enumerate}
\end{proof}

\subsection{Proof of the Theorems}\label{proofs_thms}

We first give the proof of Theorem~\ref{side}.

\begin{proof}[Proof of Theorem~\ref{side}]
In \cite{dfpp} it is shown in equations (103)--(109) that
\begin{align}\label{micro_macro}
\norm{\rho_t^{(\mathrm{micro})} - \rho_t^{(\mathrm{macro})}} &\leq  \frac{C(t)^{2}}{\Lambda }+\frac{C(t)}{\Lambda ^{1/2}}+\alpha_N(n,\Psi_N^t,\phi_t)C(t)^{2}\nonumber\\&\quad+2\sqrt{\Lambda\alpha_N(n,\Psi_N^t,\phi_t) }C(t)+\Lambda\alpha_N(n,\Psi_N^t,\phi_t)
\end{align}
for some positive function $C(t)$ under the conditions of Assumption~\ref{main_assumption}. Note that the final estimate in the proof of Theorem~1.6 in \cite{dfpp} is written down slightly different. In the inequality above we have just directly used $\norm{q_1^{\varphi_t} \Psi_N^t}^2 = \alpha_N(n,\Psi_N,\phi_t)$ in the estimates (107), (108), and (109) from \cite{dfpp}. Note that \cite{dfpp} does not track explicitly what the time-dependence of the constants in the estimates is, so let us not do that here either. Simplifying \eqref{micro_macro} and using \eqref{corzwei} from Lemma~\ref{cor} we directly get
\begin{align}
\norm{\rho_t^{(\mathrm{micro})} - \rho_t^{(\mathrm{macro})}} &\leq C'(t) \Big( \Lambda^{-1} + \Lambda^{-1/2} + \sqrt{\Lambda\alpha_N(n,\Psi_N^t,\phi_t)} + \Lambda\alpha_N(n,\Psi_N^t,\phi_t) \Big) \nonumber\\
&\leq C'(j,t) \left( \Lambda^{-1} + \Lambda^{-1/2} + \sqrt{\frac{\Lambda}{\rho} + \Lambda \left(\frac{\Lambda}{\rho}\right)^j} + \frac{\Lambda}{\rho} + \Lambda \left(\frac{\Lambda}{\rho}\right)^j \right)
\end{align}
for any $j\in\NNN$ and some positive functions $C'(t), C'(j,t)$. Assuming $\Lambda \leq C\rho^{1-\varepsilon}$ we get for $j \geq \frac{1}{\varepsilon}$ that
\begin{align}
\norm{\rho_t^{(\mathrm{micro})} - \rho_t^{(\mathrm{macro})}} &\leq C''(t) \left( \Lambda^{-1} + \Lambda^{-1/2} + \sqrt{\frac{\Lambda}{\rho}} \, \right)
\end{align}
for some positive function $C''(t)$.
\end{proof}

In order to prove Theorem~\ref{main_Bog_thm}, we first establish convergence of $\Psi_N^t$ to $\widetilde{\Psi}_N^t$, similarly as it was done in \cite{mpp}. Note that the proof is very similar to other existing derivations of Bogoliubov theory, see, e.g., \cite{Lewin:2015a,nam:2015,nam:2016}.
\begin{lemma}\label{main}
Let $\Lambda\geq 1$, $T>0$, $v \in L^1(\RRR^3) \cap L^{\infty}(\RRR^3)$ non-negative, and $1 = \big\|\Psi_N^0\big\| = \big\|\widetilde{\Psi}_N^0\big\| = \|\varphi_0\|$. Assume
\begin{equation}\label{cond_on_ini_cond_main}
\sum_{k=0}^N \left(\frac{k}{N}\right)^{\ell} \norm{P_{k}^{\varphi_0} \widetilde{\Psi}_N^0}^2 \leq c_{\ell} \rho^{-\ell}
\end{equation}
for all $1 \leq \ell \leq 4$ and for some $c_{\ell} \in \RRR$. Also assume that there is a constant $M<\infty$, uniform in $\Lambda$ and $\rho$, such that the solution $\phi_t$ of the Hartree equation \eqref{meanfield} satisfies $\sup_{0\leq t \leq T}\|\phi_t\|_\infty<M\Lambda^{-1/2}$. Let $\Psi_N^t$ be the solution to \eqref{schroe} and $\widetilde{\Psi}_N^t$ to \eqref{psitilde}. Then for all $t\leq T$,
\begin{equation}\label{tilde_estimate_main}
\norm{\Psi_N^t - \widetilde{\Psi}_N^t} \leq \norm{\Psi_N^0 - \widetilde{\Psi}_N^0} + C_4^{1/2} \big( e^{D_4t} - 1 \big) \, \sqrt{\frac{\Lambda^3}{\rho}}
\end{equation}
with $C_4$ and $D_4$ defined as in \eqref{constants_from_thm}.
\end{lemma}

\begin{proof}[Proof of Lemma~\ref{main}]

First, recall from Definition~\ref{defalpha} that
\begin{equation}
W_{j,k}^{\varphi}:=\Lambda(N-1)\bigg(v (x_j-x_k)- (v*|\phi|^2)(x_j) - (v*|\phi|^2)(x_k) + 2\mu^{\varphi}\bigg)\;.
\end{equation}
Let us write
\begin{equation}
H_N = \sum_{j=1}^N h_j^{\varphi_t} + \frac{1}{2N(N-1)} \sum_{j \neq k} W^{\varphi_t}_{j,k}.
\end{equation}
Using \eqref{projonf} it follows that $p^{\varphi}_jW^{\varphi}_{j,k}p^{\varphi}_j=0$ for any $j\neq k$, such that we have
\begin{align}
H_N-\widetilde{H}_N^{\varphi_t} &= \frac{1}{N(N-1)}\sum_{j \neq k} \bigg( \frac{1}{2}q^{\varphi_t}_jq^{\varphi_t}_kW^{\varphi_t}_{j,k}q^{\varphi_t}_jq^{\varphi_t}_k + p^{\varphi_t}_jq^{\varphi_t}_kW^{\varphi_t}_{j,k}q^{\varphi_t}_jq^{\varphi_t}_k + q^{\varphi_t}_jq^{\varphi_t}_kW^{\varphi_t}_{j,k}p^{\varphi_t}_jq^{\varphi_t}_k\bigg).
\end{align}
We now use the Duhamel expansions
\begin{align}
\Psi_N^t - \widetilde\Psi_N^t = -i \int_0^t ds\, U(t,s) \Big( H_N - \widetilde{H}_N^{\varphi_s} \Big) \widetilde\Psi_N^s = -i \int_0^t ds\, \widetilde{U}(t,s) \Big( H_N - \widetilde{H}_N^{\varphi_s} \Big) \Psi_N^s,
\end{align}
where $U(t,s)$ and $\widetilde{U}(t,s)$ are the unitary propagators generated by $H_N$ and $\widetilde{H}_N^{\varphi_t}$. Then we find
\begin{align}\label{norm_diff_3q_4q}
&\|\Psi_N^t-\widetilde\Psi_N^t\|^2 \nonumber\\
&\quad= \laa \Psi_N^t, \Psi_N^t-\widetilde{\Psi}_N^t \raa - \laa \widetilde{\Psi}_N^t, \Psi_N^t-\widetilde{\Psi}_N^t \raa \nonumber\\
&\quad= -i \int_0^t ds\, \laa \Psi_N^s, \Big( H_N - \widetilde{H}_N^{\varphi_s} \Big) \widetilde{\Psi}_N^s \raa + i \int_0^t ds\, \laa \widetilde{\Psi}_N^s, \Big( H_N - \widetilde{H}_N^{\varphi_s} \Big) \Psi_N^s \raa \nonumber\\
&\quad= 2 \int_0^t ds\, \Im \laa \Psi_N^s - \widetilde{\Psi}_N^s, \Big( H_N - \widetilde{H}_N^{\varphi_s} \Big) \widetilde{\Psi}_N^s \raa \nonumber\\
&\quad= \int_0^t ds\, \Im \laa \Psi_N^s - \widetilde{\Psi}_N^s, \Big( q^{\varphi_s}_1q^{\varphi_s}_2W^{\varphi_s}_{1,2}q^{\varphi_s}_1q^{\varphi_s}_2 + 2p^{\varphi_s}_1q^{\varphi_s}_2\widetilde{W}^{\varphi_s}_{1,2}q^{\varphi_s}_1q^{\varphi_s}_2 + 2q^{\varphi_s}_1q^{\varphi_s}_2\widetilde{W}^{\varphi_s}_{1,2}p^{\varphi_s}_1q^{\varphi_s}_2\Big) \widetilde{\Psi}_N^s \raa,
\end{align}
where, as before, we set
\begin{equation}
\widetilde{W}_{j,k}^{\varphi} := \Lambda(N-1)\bigg(v (x_j-x_k)- (v*|\phi|^2)(x_j) - (v*|\phi|^2)(x_k)\bigg).
\end{equation}
In order to control the terms with one $p^{\varphi_s}$ projector, we introduce for $j>0$, (with a slight abuse of notation) $\widehat{n}^{-j} := \sum_{k=1}^N \left(\frac{N}{k}\right)^{j} P_k^{\varphi}$ for any normalized $\varphi \in L^2(\RRR^3)$. Here, we omit the superscript $\varphi$ for the $\,\widehat{\cdot}\,$ operators for ease of notation. Then, for any $\Phi_N,\widetilde{\Psi}_N\in L^2(\RRR^{3N})$, we find
\begin{align}
&\laa\Phi_N, p^{\varphi}_1q^{\varphi}_2 \widetilde{W}^{\varphi}_{1,2} q^{\varphi}_1q^{\varphi}_2 \widetilde{\Psi}_N\raa \nonumber\\
&\qquad= \laa\Phi_N, p^{\varphi}_1q^{\varphi}_2 \widehat{n}^{-1/2}\widehat{n}^{1/2} \widetilde{W}^{\varphi}_{1,2} q^{\varphi}_1q^{\varphi}_2 \widetilde{\Psi}_N\raa \nonumber\\
&\qquad= \laa\Phi_N, \widehat{n}^{-1/2} p^{\varphi}_1q^{\varphi}_2 \widetilde{W}^{\varphi}_{1,2} q^{\varphi}_1q^{\varphi}_2 \widehat{n}^{1/2}_{-1} \widetilde{\Psi}_N\raa \nonumber\\
&\qquad\leq \norm[op]{\widetilde{W}^{\varphi}_{1,2}p_1^{\varphi}} \norm{q^{\varphi}_2 \widehat{n}^{-1/2}\Phi_N} \norm{q^{\varphi}_1q^{\varphi}_2 \widehat{n}^{1/2}_{-1} \widetilde{\Psi}_N} \nonumber\\
&\qquad= \norm[op]{\widetilde{W}^{\varphi}_{1,2}p_1^{\varphi}} \bigscp{\Phi_N}{q^{\varphi}_2\widehat{n}^{-1}\Phi_N}^{1/2} \bigscp{\widetilde{\Psi}_N}{q^{\varphi}_1q^{\varphi}_2\widehat{n}_{-1}\widetilde{\Psi}_N}^{1/2} \nonumber\\
&\qquad= \norm[op]{\widetilde{W}^{\varphi}_{1,2}p_1^{\varphi}} \bigscp{\Phi_N}{\sum_{k=1}^N \frac{k}{N} \frac{N}{k} P_k^{\varphi}\Phi_N}^{1/2} \bigscp{\widetilde{\Psi}_N}{\sum_{k=2}^N \frac{k(k-1)}{N(N-1)} \frac{(k-1)}{N} P_k^{\varphi}\widetilde{\Psi}_N}^{1/2} \nonumber\\
&\qquad\leq \norm[op]{\widetilde{W}^{\varphi}_{1,2}p_1^{\varphi}} \norm{\Phi_N} \Big(2 \alpha_N(n^3,\widetilde{\Psi}_N,\varphi)\Big)^{1/2}.
\end{align}
In a similar way, we find
\begin{align}
&\laa\Phi_N, q^{\varphi}_1q^{\varphi}_2 \widetilde{W}^{\varphi}_{1,2} p^{\varphi}_1q^{\varphi}_2 \widetilde{\Psi}_N\raa \nonumber\\
&\qquad= \laa\Phi_N, q^{\varphi}_1q^{\varphi}_2 \widetilde{W}^{\varphi}_{1,2} p^{\varphi}_1q^{\varphi}_2 \widehat{n}^{-1}\widehat{n}\widetilde{\Psi}_N\raa \nonumber\\
&\qquad= \laa\Phi_N, q^{\varphi}_1q^{\varphi}_2 \widehat{n}^{-1}_{-1} \widetilde{W}^{\varphi}_{1,2} p^{\varphi}_1q^{\varphi}_2 \widehat{n} \widetilde{\Psi}_N\raa \nonumber\\
&\qquad\leq \norm[op]{\widetilde{W}^{\varphi}_{1,2}p_1^{\varphi}} \norm{q^{\varphi}_1q^{\varphi}_2 \widehat{n}^{-1}_{-1}\Phi_N} \norm{q^{\varphi}_2 \widehat{n} \widetilde{\Psi}_N} \nonumber\\
&\qquad= \norm[op]{\widetilde{W}^{\varphi}_{1,2}p_1^{\varphi}} \bigscp{\Phi_N}{q^{\varphi}_1q^{\varphi}_2\widehat{n}^{-2}_{-1}\Phi_N}^{1/2} \bigscp{\widetilde{\Psi}_N}{q^{\varphi}_2\widehat{n}^2 \widetilde{\Psi}_N}^{1/2} \nonumber\\
&\qquad= \norm[op]{\widetilde{W}^{\varphi}_{1,2}p_1^{\varphi}} \bigscp{\Phi_N}{\sum_{k=2}^N \frac{k(k-1)}{N(N-1)} \frac{N^2}{(k-1)^2} P_k^{\varphi}\Phi_N}^{1/2} \bigscp{\widetilde{\Psi}_N}{\sum_{k=1}^N \frac{k}{N} \frac{k^2}{N^2} P_k^{\varphi}\widetilde{\Psi}_N}^{1/2} \nonumber\\
&\qquad\leq \norm[op]{\widetilde{W}^{\varphi}_{1,2}p_1^{\varphi}} 2 \norm{\Phi_N} \, \alpha_N(n^3,\widetilde{\Psi}_N,\varphi)^{1/2}.
\end{align}
For controlling the term with four $q^{\varphi_s}$ projectors in \eqref{norm_diff_3q_4q} we use
\begin{align}
&\laa \Phi_N, q^{\varphi}_1q^{\varphi}_2W^{\varphi}_{1,2}q^{\varphi}_1q^{\varphi}_2 \widetilde{\Psi}_N \raa \nonumber\\
&\qquad= \laa \Phi_N, q^{\varphi}_1q^{\varphi}_2W^{\varphi}_{1,2}q^{\varphi}_1q^{\varphi}_2 \widehat{n}^{-1} \widehat{n} \widetilde{\Psi}_N \raa \nonumber\\
&\qquad= \laa \Phi_N, \widehat{n}^{-1} q^{\varphi}_1q^{\varphi}_2W^{\varphi}_{1,2}q^{\varphi}_1q^{\varphi}_2 \widehat{n} \widetilde{\Psi}_N \raa \nonumber\\
&\qquad\leq \norm[op]{W^{\varphi}_{1,2}} \norm{q^{\varphi}_1q^{\varphi}_2\widehat{n}^{-1}\Phi_N} \norm{q^{\varphi}_1q^{\varphi}_2 \widehat{n} \widetilde{\Psi}_N} \nonumber\\
&\qquad= \norm[op]{W^{\varphi}_{1,2}} \bigscp{\Phi_N}{q^{\varphi}_1q^{\varphi}_2\widehat{n}^{-2}\Phi_N}^{1/2} \bigscp{\widetilde{\Psi}_N}{q^{\varphi}_1q^{\varphi}_2\widehat{n}^2 \widetilde{\Psi}_N}^{1/2} \nonumber\\
&\qquad= \norm[op]{W^{\varphi}_{1,2}} \bigscp{\Phi_N}{\sum_{k=2}^N \frac{k(k-1)}{N(N-1)} \frac{N^2}{k^2} P_k^{\varphi} \Phi_N}^{1/2} \bigscp{\widetilde{\Psi}_N}{\sum_{k=2}^N \frac{k(k-1)}{N(N-1)} \frac{k^2}{N^2} P_k^{\varphi} \widetilde{\Psi}_N}^{1/2} \nonumber\\
&\qquad\leq \norm[op]{W^{\varphi}_{1,2}} \sqrt{2} \norm{\Phi_N} \alpha_N(n^4,\widetilde{\Psi}_N,\varphi)^{1/2}.
\end{align}
It follows that
\begin{align}
&\|\Psi_N^t-\widetilde\Psi_N^t\|^2 \nonumber\\
&\quad\leq \Big|\int_0^t ds\, \laa\Psi_N^s - \widetilde\Psi_N^s, q^{\varphi_s}_1q^{\varphi_s}_2W^{\varphi_s}_{1,2}q^{\varphi_s}_1q^{\varphi_s}_2\widetilde\Psi_N^s\raa \Big| \nonumber\\
&\qquad + 2 \Big|\int_0^t ds\, \laa \Psi_N^s-\widetilde{\Psi}_N^s, \left( p^{\varphi_s}_1q^{\varphi_s}_2\widetilde{W}^{\varphi_s}_{1,2}q^{\varphi_s}_1q^{\varphi_s}_2 + q^{\varphi_s}_1q^{\varphi_s}_2\widetilde{W}^{\varphi_s}_{1,2}p^{\varphi_s}_1q^{\varphi_s}_2\right)  \widetilde\Psi_N^s\raa \Big| \nonumber\\
&\quad\leq 8 \int_0^t ds\, \norm{\Psi_N^s-\widetilde{\Psi}_N^s} \bigg(\|W^{\varphi_s}_{1,2}\|_{op} \alpha_N(n^4,\widetilde{\Psi}_N^s,\varphi_s)^{1/2} + \|\widetilde{W}^{\varphi_s}_{1,2}p^{\varphi_s}_1\|_{op} \alpha_N(n^3,\widetilde{\Psi}_N^s,\varphi_s)^{1/2} \bigg).
\end{align}
Let us consider $\Psi_N^0 = \widetilde{\Psi}_N^0$ first. Then this inequality implies that
\begin{align}
\norm{\Psi_N^t - \widetilde{\Psi}_N^t} &\leq 4 \int_0^t ds\, \bigg( \|W^{\varphi_s}_{1,2}\|_{op} \alpha_N(n^4,\widetilde{\Psi}_N^s,\varphi_s)^{1/2} + \|\widetilde{W}^{\varphi_s}_{1,2}p^{\varphi_s}_1\|_{op} \alpha_N(n^3,\widetilde{\Psi}_N^s,\varphi_s)^{1/2} \bigg).
\end{align}
Using \eqref{cordrei} from Lemma~\ref{cor} with increasing constants $\widetilde{C}_j$ and $D_j$, $\norm[op]{W_{1,2}^{\varphi_s}} \leq 4 \Lambda N \norm[\infty]{v}$ and $\|\widetilde{W}^{\varphi_s}_{1,2}p^{\varphi_s}_1\|_{op} \leq 3\Lambda N \norm[2]{v} \norm[\infty]{\varphi_s}$ we find,
\begin{align}
\norm{\Psi_N^t - \widetilde{\Psi}_N^t} &\leq 4 \int_0^t ds\, \bigg( 4 \Lambda N \|v\|_\infty \big(\widetilde{C}_4 e^{D_4s}\rho^{-4}\big)^{1/2} + 3\Lambda N \|v\|_2 M \Lambda^{-1/2} \big(\widetilde{C}_3 e^{D_3s}\rho^{-3}\big)^{1/2} \bigg) \nonumber\\
&\leq \int_0^t ds\, 8 \Big( 4 \|v\|_\infty + 3 M \|v\|_2 \Big) \big(\widetilde{C}_4 e^{D_4s}\big)^{1/2} \Lambda^{3/2} \rho^{-1/2} \nonumber\\
&\leq \widetilde{C}_4^{1/2} \big( e^{D_4t} - 1 \big) \Lambda^{3/2} \rho^{-1/2}.
\end{align}
If $\Psi_N^0 \neq \widetilde{\Psi}_N^0$ we use unitarity of the propagator and the triangle inequality, i.e.,
\begin{align}
\norm{U(t,0)\Psi_N^0 - \widetilde{U}(t,0)\widetilde{\Psi}_N^0} &= \norm{U(t,0)\Psi_N^0 - U(t,0)\widetilde{\Psi}_N^0 + U(t,0)\widetilde{\Psi}_N^0 - \widetilde{U}(t,0)\widetilde{\Psi}_N^0} \nonumber\\
&\leq \norm{\Psi_N^0 - \widetilde{\Psi}_N^0} + \norm{U(t,0)\widetilde{\Psi}_N^0 - \widetilde{U}(t,0)\widetilde{\Psi}_N^0},
\end{align}
which proves \eqref{tilde_estimate_main}.
\end{proof}

Next, we prove convergence of $\widetilde{\Psi}_N^t$ to the solution of the Bogoliubov equation \eqref{Bog_eq}. As preparation, let us decompose a solution $\widetilde{\Psi}_N^t$ to the equation $i\partial_t \widetilde{\Psi}_N^t = \widetilde{H}_N^{\varphi_t} \widetilde{\Psi}_N^t$ as in \eqref{decomp}. Then a direct calculation shows that the equation the corresponding $\widetilde{\chi}^{(k)}_t$'s satisfy is for all $0\leq k \leq N$ given by
\begin{align}\label{chi_k}
&i\partial_t \widetilde{\chi}^{(k)}_t(x_1,\ldots,x_k) \nonumber\\
&\quad= \sum_{i=1}^k \bigg(h^{\varphi_t}_i + \frac{N-k}{N} K^{(1)}_t(x_i)\bigg)\widetilde{\chi}^{(k)}_t \nonumber\\
&\qquad + \frac{1}{2}\frac{\sqrt{(N-k+2)(N-k+1)}}{N} \frac{1}{\sqrt{k(k-1)}} \sum_{1 \leq i < j \leq k} K^{(2)}_t(x_i,x_j) \widetilde{\chi}^{(k-2)}_t(x_1,\ldots,x_k \setminus x_i \setminus x_j) \nonumber\\
&\qquad + \frac{1}{2}\frac{\sqrt{(N-k)(N-k-1)}}{N} \sqrt{(k+1)(k+2)} \int dx \,dy\, \overline{K^{(2)}_t(x,y)} \widetilde{\chi}^{(k+2)}_t(x_1,\ldots,x_k,x,y).
\end{align}
In the equation for $k=0$ the first and second line in \eqref{chi_k} are understood to be zero, and so are for $k=1$ the second line, and for $k=N-1$ and $k=N$ the last line. Note that there are two differences between the equations \eqref{chi_k} and the Bogoliubov equations \eqref{chi_Bog_k}. One is the combinatorial factors in front of the $K^{(1)}_t$ and $K^{(2)}_t$ terms. Since the combinatorial factors are approximately given by $1 + \frac{k}{N}$ and we already proved that the expectation value of $\frac{k}{N}$ is small, we can control this error term. The other difference is that the the equations for $k\geq N-1$ are different. While $\tilde{\chi}_t^{(k)} = 0$ for all $k > N$, the $\chi_t^{\Bog,(k)}$ generally do not vanish for $k>N$. However, for $k>N$ the norm of $\chi_t^{\Bog,(k)}$ will be very small, since the norm of $\chi_t^{\Bog,(k)}$ for all $k$ of order $N$ is very small. Thus, the error terms coming from those two sources are small which allows us to prove the following lemma, similarly as it was done in \cite{mpp}. As mentioned above, the proof is very similar to other existing derivations of Bogoliubov theory, see, e.g., \cite{Lewin:2015a,nam:2015,nam:2016}.

\begin{lemma}\label{lemma_Bog}
Let $\Lambda\geq 1$, $T>0$, $v \in L^1(\RRR^3) \cap L^{\infty}(\RRR^3)$ non-negative, and $1 = \big\| \chi_0^{\Bog} \big\| = \|\varphi_0\|$. We set
\begin{align}\label{epsilon_0}
\sum_{k=N+1}^{\infty} \big\| \chi_0^{\Bog,(k)} \big\|^2 := \varepsilon_0,
\end{align}
and $\widetilde{\Psi}_N^0 = (1-\varepsilon_0)^{-1/2}\sum_{k=0}^N \varphi^{\otimes (N-k)}_0 \otimes_s \chi^{\Bog,(k)}_0$, such that $\|\widetilde{\Psi}_N^0\| = 1$. Assume
\begin{equation}\label{cond_on_ini_cond_main_in_lemma}
\sum_{k=0}^N \left(\frac{k}{N}\right)^{\ell} \norm{P_{k}^{\varphi_0} \widetilde{\Psi}_N^0}^2 \leq c_{\ell} \rho^{-\ell}
\end{equation}
for all $1 \leq \ell \leq 4$ and for some $c_{\ell} \in \RRR$. Also assume that there is a constant $M<\infty$, uniform in $\Lambda$ and $\rho$, such that the solution $\phi_t$ of the Hartree equation \eqref{meanfield} satisfies $\sup_{0\leq t \leq T}\|\phi_t\|_\infty<M\Lambda^{-1/2}$. Let $\widetilde{\Psi}_N^t$ be the solution to \eqref{psitilde} and $\chi^{\Bog}_t$ to \eqref{Bog_eq}. Then for all $t\leq T$,
\begin{equation}\label{Bog_bound}
\bigg\|\widetilde{\Psi}_N^t - \sum_{k=0}^N \varphi_t^{\otimes (N-k)} \otimes_s \chi^{\Bog,(k)}_t\bigg\| \leq \sqrt{2\varepsilon_0} + \widetilde{C}_4^{1/2} \big( e^{D_4t} - 1 \big) \frac{\Lambda}{\rho}
\end{equation}
with $C_4$ and $D_4$ defined as in \eqref{constants_from_thm}.
\end{lemma}

\begin{remark}
Note that the error term for the norm difference of $\widetilde{\Psi}_N^t$ and the Bogoliubov state is small if $\Lambda/\rho$ is small, i.e., it is much less restrictive than the main error term of order $(\Lambda^3/\rho)^{1/2}$ for the norm difference between $\Psi^t_N$ and $\widetilde{\Psi}_N^t$ in Lemma~\ref{main}.
\end{remark}

\begin{proof}[Proof of Lemma~\ref{lemma_Bog}]
Let us first note that
\begin{align}\label{tilde_Bog_norm}
\bigg\|\widetilde{\Psi}_N^t - \sum_{k=0}^N \varphi_t^{\otimes (N-k)} \otimes_s \chi^{\Bog,(k)}_t\bigg\|^2_{L^2(\RRR^{3N})} &= \bigg\|\sum_{k=0}^N \varphi_t^{\otimes (N-k)} \otimes_s \Big(\chi^{\Bog,(k)}_t - \tilde{\chi}_t^{(k)}\Big)\bigg\|^2_{L^2(\RRR^{3N})} \nonumber\\
&= \sum_{k=0}^N \Big\| \chi^{\Bog,(k)}_t - \tilde{\chi}_t^{(k)} \Big\|_{L^2(\RRR^{3k})}^2 \nonumber\\
&\leq \sum_{k=0}^\infty \Big\| \chi^{\Bog,(k)}_t - \tilde{\chi}_t^{(k)} \Big\|_{L^2(\RRR^{3k})}^2 \nonumber \\
&= 2 - 2 \, \Re \sum_{k=0}^N \bigscp{\chi^{\Bog,(k)}_t}{\tilde{\chi}^{(k)}_t}.
\end{align}
A direct calculation yields
\small
\begin{subequations}
\begin{align}
&\frac{d}{dt} \sum_{k=0}^\infty \norm[L^2(\RRR^{3k})]{\chi^{\Bog,(k)}_t - \tilde{\chi}_t^{(k)}}^2 \nonumber\\
\label{Bog_est_1}
&\quad = 2 \sum_{k=1}^N \Im \bigscp{\chi^{\Bog,(k)}_t - \tilde{\chi}_t^{(k)}}{\left( 1-\frac{N-k}{N}\right) \int \! dx\, K^{(1)}_t(x) \big(a^\dagger_x a_x\tilde{\chi}_t\big)^{(k)}} \\
\label{Bog_est_2}
&\qquad + \sum_{k=2}^N \Im \bigscp{\chi^{\Bog,(k)}_t - \tilde{\chi}_t^{(k)}}{\left( 1 - \frac{\sqrt{(N-k+2)(N-k+1)}}{N}\right) \int\! dx \! \int\! dy\, K^{(2)}_t(x,y) \big(a_x^\dagger a_y^\dagger \tilde{\chi}_t\big)^{(k)}} \\
\label{Bog_est_3}
&\qquad + \sum_{k=0}^{N-2} \Im \bigscp{\chi^{\Bog,(k)}_t - \tilde{\chi}_t^{(k)}}{\left(1 - \frac{\sqrt{(N-k)(N-k-1)}}{N}\right) \int\! dx \! \int\! dy\, \overline{K^{(2)}_t(x,y)} \big(a_x a_y\tilde{\chi}_t\big)^{(k)}} \\
\label{Bog_est_4}
&\qquad + \sum_{k=N-1}^N \Im \bigscp{\int\! dx \! \int\! dy\, \overline{K^{(2)}_t(x,y)} \big(a_x a_y\tilde{\chi}^{\Bog}_t\big)^{(k)}}{\tilde{\chi}_t^{(k)}}.
\end{align}
\end{subequations}
\normalsize
The first three terms come from the difference in the combinatorial factors in the equations for $\tilde{\chi}_t^{(k)}$ and $\chi^{\Bog,(k)}_t$, and the last term comes from the fact that $\chi^{\Bog,(k)}_t$ can have non-zero $k$-particle sectors for $k>N$. To estimates the first three terms, note that the combinatorial factors are all bounded by $\frac{k+1}{N}$, and also $K^{(1)}_t$ and $K^{(2)}_t$ are bounded, more precisely,
\begin{align}
\big\|K^{(1)}_t\big\|_{\Hilbert} \leq \Lambda \norm[\infty]{\varphi_t}^2 \norm[1]{v} \leq M \norm[1]{v},
\end{align}
and
\begin{align}
\big\|K^{(2)}_t\big\|_{\Hilbert^{\otimes 2}} \leq \Lambda \norm[\infty]{\varphi_t}^2 \norm[2]{v} \leq M \norm[2]{v}.
\end{align}
Then we can use that for any $\big(\chi^{(k)}\big)_k$,
\begin{align}
\norm{\int \! dx\, K^{(1)}_t(x) \big(a^\dagger_x a_x\chi\big)^{(k)}} &\leq k \big\|K^{(1)}_t\big\|_{\Hilbert} \norm{\chi^{(k)}} \leq M \norm[1]{v} k \norm{\chi^{(k)}}, \\
\norm{\int\! dx \! \int\! dy\, K^{(2)}_t(x,y) \big(a_x^\dagger a_y^\dagger \chi\big)^{(k)}} &\leq \sqrt{k(k-1)} \big\|K^{(2)}_t\big\|_{\Hilbert^{\otimes 2}} \norm{\chi^{(k-2)}}\nonumber\\ 
&\leq M \norm[2]{v} \sqrt{k(k-1)} \norm{\chi^{(k-2)}},
\end{align}
and
\begin{align}
\norm{\int\! dx \! \int\! dy\, \overline{K^{(2)}_t(x,y)} \big(a_x a_y \chi\big)^{(k)}} \leq& \sqrt{(k+1)(k+2)} \big\|K^{(2)}_t\big\|_{\Hilbert^{\otimes 2}} \norm{\chi^{(k+2)}} \nonumber\\\leq & M \norm[2]{v} \sqrt{(k+1)(k+2)}\norm{\chi^{(k+2)}}.
\end{align}
So for the first three terms in the derivative of the norm difference we find the bounds
\begin{align}
|\eqref{Bog_est_1}| &\leq 2 M \norm[1]{v} \sum_{k=0}^N \Big\| \chi^{\Bog,(k)}_t - \tilde{\chi}_t^{(k)} \Big\| \frac{k^2}{N}\norm{\tilde{\chi}^{(k)}_t} \nonumber \\
& \leq 2 M \norm[1]{v} \sqrt{\sum_{k=0}^{\infty} \Big\| \chi^{\Bog,(k)}_t - \tilde{\chi}_t^{(k)}\Big\|^2} \sqrt{N^2 \alpha_N(n^4,\widetilde{\Psi}_N^t,\varphi_t)}, \\
|\eqref{Bog_est_2}| &\leq M \norm[2]{v} \sum_{k=2}^N \Big\| \chi^{\Bog,(k)}_t - \tilde{\chi}_t^{(k)} \Big\| \frac{\sqrt{k(k-1)^3}}{N}\norm{\tilde{\chi}^{(k-2)}_t}, \nonumber\\
& \leq 4 M \norm[2]{v} \sqrt{\sum_{k=0}^{\infty} \Big\| \chi^{\Bog,(k)}_t - \tilde{\chi}_t^{(k)}\Big\|^2} \sqrt{N^2 \alpha_N(n^4,\widetilde{\Psi}_N^t,\varphi_t) + N^{-2}}, \\
|\eqref{Bog_est_3}| &\leq M \norm[2]{v} \sum_{k=0}^{N-2} \Big\| \chi^{\Bog,(k)}_t - \tilde{\chi}_t^{(k)} \Big\| \frac{\sqrt{(k+1)^3(k+2)}}{N}\norm{\tilde{\chi}^{(k+2)}_t} \nonumber\\
& \leq M \norm[2]{v} \sqrt{\sum_{k=0}^{\infty} \Big\| \chi^{\Bog,(k)}_t - \tilde{\chi}_t^{(k)}\Big\|^2} \sqrt{N^2 \alpha_N(n^4,\widetilde{\Psi}_N^t,\varphi_t)}.
\end{align}
The term \eqref{Bog_est_4} can be estimated by
\begin{align}
|\eqref{Bog_est_4}| &\leq \sum_{k=N-1}^N M \norm[2]{v} (k+2) \Big\| \chi^{\Bog,(k+2)}_t \Big\| \, \Big\| \tilde{\chi}^{(k)}_t \Big\| \nonumber\\
&\leq 4 M\norm[2]{v}(N+2) \sum_{k=N-1}^N \Big\| \chi^{\Bog,(k+2)}_t \Big\| \, \Big\| \mathbbm{1}(k\leq N) \frac{k^2}{N^2} \tilde{\chi}^{(k)}_t \Big\| \nonumber\\
&\leq 8 M\norm[2]{v} \sqrt{\sum_{k=0}^{\infty} \Big\| \chi^{\Bog,(k)}_t - \tilde{\chi}_t^{(k)}\Big\|^2} \sqrt{N^2 \alpha_N(n^4,\widetilde{\Psi}_N^t,\varphi_t)}.
\end{align}
Summarizing our estimates and using \eqref{cordrei} from Lemma~\ref{cor}, we find
\begin{align}
&\frac{d}{dt} \sum_{k=0}^\infty \norm[L^2(\RRR^{3k})]{\chi^{\Bog,(k)}_t - \tilde{\chi}_t^{(k)}}^2 \nonumber\\
&\quad \leq 13 M \big( \norm[1]{v} + \norm[2]{v} \big) \sqrt{\sum_{k=0}^{\infty} \Big\| \chi^{\Bog,(k)}_t - \tilde{\chi}_t^{(k)}\Big\|^2} \sqrt{N^2 \alpha_N(n^4,\widetilde{\Psi}_N^t,\varphi_t) + N^{-2}} \nonumber\\
& \quad \leq 13 M \big( \norm[1]{v} + \norm[2]{v} \big) \sqrt{\sum_{k=0}^{\infty} \Big\| \chi^{\Bog,(k)}_t - \tilde{\chi}_t^{(k)}\Big\|^2} \sqrt{2 \widetilde{C}_4 e^{D_4t} \frac{\Lambda^2}{\rho^2}}.
\end{align}
Integrating and putting together the constants give
\begin{equation}
\sqrt{\sum_{k=0}^\infty \norm[L^2(\RRR^{3k})]{\chi^{\Bog,(k)}_t - \tilde{\chi}_t^{(k)}}^2} \leq \sqrt{\sum_{k=0}^\infty \norm[L^2(\RRR^{3k})]{\chi^{\Bog,(k)}_0 - \tilde{\chi}_0^{(k)}}^2} + \sqrt{\widetilde{C}_4} \Big( e^{D_4t} - 1 \Big) \frac{\Lambda}{\rho}.
\end{equation}
Finally, note that with the notation from \eqref{epsilon_0} we have
\begin{align}\label{Bog_tilde_epsilon}
\sum_{k=0}^\infty \norm[L^2(\RRR^{3k})]{\chi^{\Bog,(k)}_0 - \tilde{\chi}_0^{(k)}}^2 &= \sum_{k=0}^N \norm[L^2(\RRR^{3k})]{\chi^{\Bog,(k)}_0 - \tilde{\chi}_0^{(k)}}^2 + \sum_{k=N+1}^{\infty} \norm[L^2(\RRR^{3k})]{\chi^{\Bog,(k)}_0}^2 \nonumber\\
&= \left( (1-\varepsilon_0)^{-1/2} - 1 \right)^2 (1-\varepsilon_0) + \varepsilon_0 \nonumber\\
&\leq 2 \varepsilon_0,
\end{align}
which proves \eqref{Bog_bound}.
\end{proof}

By the triangle inequality, we directly arrive at a bound for the norm difference of $\Psi^t_N$ and the Bogoliubov state.

\begin{proof}[Proof of Theorem~\ref{main_Bog_thm}]
As in Lemma~\ref{lemma_Bog}, let us define $\widetilde{\Psi}_N^0 = (1-\varepsilon_0)^{-1/2}\sum_{k=0}^N \varphi^{\otimes (N-k)}_0 \otimes_s \chi^{\Bog,(k)}_0$ with $\varepsilon_0 := \sum_{k=N+1}^{\infty} \big\| \chi_0^{\Bog,(k)} \big\|^2$, such that $\|\widetilde{\Psi}_N^0\| = 1$, and let  $\widetilde{\Psi}_N^t$ be the solution to \eqref{psitilde} with initial condition $\widetilde{\Psi}_N^0$. In order to express $\varepsilon_0$ in terms of the initial norm difference, we estimate
\begin{equation}\label{epsilon_ini_norm}
\bigg\|\Psi_N^0 - \sum_{k=0}^N \varphi_0^{\otimes (N-k)} \otimes_s \chi^{\Bog,(k)}_0\bigg\|^2 \geq 2(1-\sqrt{1-\varepsilon_0}) - \varepsilon_0 \geq \frac{\varepsilon_0^2}{4}. 
\end{equation}
Let us assume for the moment that 
\begin{equation}\label{norm_assumption}
\bigg\|\Psi_N^0 - \sum_{k=0}^N \varphi_0^{\otimes (N-k)} \otimes_s \chi^{\Bog,(k)}_0\bigg\| \leq \frac{1}{4},
\end{equation}
such that $\varepsilon_0 \leq \frac{1}{2}$ Then the condition \eqref{cond_on_ini_cond_main} from Lemma~\ref{main} holds, since
\begin{align}
\sum_{k=0}^N \left(\frac{k}{N}\right)^{\ell} \norm{P_{k}^{\varphi_0} \widetilde{\Psi}_N^0}^2 &= \sum_{k=0}^N \left(\frac{k}{N}\right)^{\ell} (1-\varepsilon_0)^{-1} \norm{\chi^{\Bog,(k)}_0}^2 \leq (1-\varepsilon_0)^{-1} c_{\ell} \rho^{-\ell} \leq 2 c_{\ell} \rho^{-\ell}
\end{align}
due to the assumption \eqref{cond_on_ini_cond_main_in_thm}. Then we get with the triangle inequality, Lemma~\ref{main}, and Lemma~\ref{lemma_Bog} that
\begin{align}
\bigg\|\Psi_N^t - \sum_{k=0}^N \varphi_t^{\otimes (N-k)} \otimes_s \chi^{\Bog,(k)}_t\bigg\| &\leq \bigg\|\Psi_N^t - \widetilde{\Psi}_N^t \bigg\| + \bigg\|\widetilde{\Psi}_N^t - \sum_{k=0}^N \varphi_t^{\otimes (N-k)} \otimes_s \chi^{\Bog,(k)}_t\bigg\| \nonumber\\
&\leq \bigg\|\Psi_N^0 - \widetilde{\Psi}_N^0 \bigg\| + 2 C(t)\sqrt{\frac{\Lambda^3}{\rho}} + \sqrt{2\varepsilon_0} + 2 C(t) \frac{\Lambda}{\rho}.
\end{align}
Note that due to \eqref{tilde_Bog_norm} and \eqref{Bog_tilde_epsilon},
\begin{align}
\bigg\|\Psi_N^0 - \widetilde{\Psi}_N^0 \bigg\| \leq \bigg\|\Psi_N^0 - \sum_{k=0}^N \varphi_0^{\otimes (N-k)} \otimes_s \chi^{\Bog,(k)}_0\bigg\| + \sqrt{2\varepsilon_0}.
\end{align}
Then \eqref{Bog_bound_psi_thm} follows from using \eqref{epsilon_ini_norm} and $\Lambda \geq 1$. Inequality \eqref{Bog_bound_psi_thm} trivially also holds when the initial norm difference in \eqref{norm_assumption} is bigger than $1/4$.
\end{proof}

\noindent{\textbf{Acknowledgments.}} We are grateful for the hospitality of the Institute for Advanced Study and Central China Normal University (CCNU), where parts of this work were done. The comments of the referees, which greatly improved this article, are especially acknowledged. We would like to thank Phan Th\`{a}nh Nam for helpful comments on the article and interesting discussions that led to Equation~\eqref{explicit_eta_solution}. S.\,P.\ gratefully acknowledges support from the German Academic Exchange Service (DAAD) and the National Science Foundation under agreement No.\ DMS-1128155, and thanks the University of Washington for hospitality. A.\,S.\ is partially supported by NSF DMS grant 01600749 and CNSF grant 11671163.

\end{document}